\documentclass[fontsize=10pt]{article}
\usepackage[hyphens]{url}
\usepackage[a4paper, total={6in, 8in}]{geometry}
\usepackage{graphicx}
\usepackage{natbib}
\usepackage[font=small,labelfont=bf]{caption}

\usepackage{mathtools}
\usepackage{amsthm}
\usepackage{amsmath}
\usepackage{amsfonts}
\usepackage{amssymb}

\usepackage{hyperref}
\usepackage{cleveref}

\usepackage{algorithm}
\usepackage{algorithmic}

\usepackage{authblk}

\usepackage{preamble}

\title{The Value of Recall in Extensive-Form Games}
\author[1,2]{Ratip Emin Berker}
\author[1,2]{Emanuel Tewolde}
\author[1]{Ioannis Anagnostides}
\author[1,4]{\\ Tuomas Sandholm}
\author[1,2,3]{Vincent Conitzer}

\affil[1]{Carnegie Mellon University}
\affil[2]{Foundations of Cooperative AI Lab (FOCAL)}
\affil[3]{University of Oxford}
\affil[4]{Strategic Machine, Inc.}
\affil[5]{Strategy Robot, Inc.}
\affil[6]{Optimized Markets, Inc.}
\affil[ ]{\texttt{\{rberker, etewolde, ianagnos, sandholm, conitzer\}@cs.cmu.edu}}

\begin{document}

\maketitle

\begin{abstract}
    Imperfect-recall games---in which players may forget previously acquired information---have found many practical applications, ranging from game abstractions to team games and testing AI agents. In this paper, we quantify the utility gain by endowing a player with perfect recall, which we call the \emph{value of recall (VoR)}. While VoR can be unbounded in general, we parameterize it in terms of various game properties, namely the structure of chance nodes and the \emph{degree of absentmindedness} (the number of successive times a player enters the same information set). Further, we identify several pathologies that arise with VoR, and show how to circumvent them. We also study the complexity of computing VoR, and how to optimally apportion \emph{partial recall}. Finally, we connect VoR to other previously studied concepts in game theory, including the price of anarchy. We use that connection in conjunction with the celebrated \emph{smoothness} framework to characterize VoR in a broad class of games.
\end{abstract}

\section{Introduction}

Game theory offers a principled framework for reasoning about complex interactions that involve multiple strategic players. It continues to propel landmark results in long-standing challenges in artificial intelligence (AI), ranging from poker~\citep{Brown17:Superhuman,Bowling15:Heads,Moravvcik17:DeepStack} to diplomacy~\citep{Bakhtin22:Human}. A common premise in game-theoretic modeling is \emph{perfect recall}---players never forget information once acquired. The perfect-recall assumption is often called into question for games involving human players; however, it is difficult to come up with a faithful model in such cases due to the unpredictability of when and what human players will forget. In contrast, AI agents can be specifically designed to relinquish certain information, thereby making the imperfect-recall framework directly applicable. But why should one consider AI agents with imperfect recall?

An early, influential application of imperfect-recall games revolves around \emph{abstraction}: games encountered in practice are typically too large to represent exactly, and so one resorts to abstraction to compress its description. In particular, one way of doing so consists of allowing players to carefully forget less important aspects of previously held information. Indeed, imperfect-recall abstractions have been a crucial component of state-of-the-art algorithms in poker solving~\citep{Brown15:Hierarchical,Johanson13:Evaluating,Waugh09:Practical,Ganzfried14:Potential,Cermak17:Algorithm}. Imperfect recall also naturally arises in so-called \emph{adversarial team games}~\citep{Celli18:Computational,Zhang23:Team,Zhang22:Correlation,Emmons22:Learning,VonStengel97:Team}, wherein a team of players---which can be construed as a single player with imperfect recall---faces an adversary. The benefit of reinforcing the communication capacity of the team in such settings---corresponding to boosting recall---is an active area of research, prominently featured in a recent NeurIPS competition~\citep{Meisheri20:Sample,Resnick2020:Pommerman}. Relatedly, natural notions of \emph{correlated equilibria} can be modeled via an imperfect-recall mediator, endowed with the ability to provide recommendations~\citep{Zhang22:Polynomial}; in that context, imperfect recall can serve to safeguard players' private information, a consideration that also arises in other settings~\citep{Conitzer19:Designing}. Finally, another possible application revolves around simulating and testing AI agents before their deployment in the real world \citep{Kovarik23:Game,Kovarik24:Recursive,Chen2024:Imperfect}. As a result, it is becoming increasingly pressing to expand our scope beyond the assumption of perfect recall.

In this paper, we examine a question at the heart of this research agenda: \emph{how does perfect recall affect players' utilities under various natural solution concepts?} More specifically, we contrast the utilities obtained by a player in an initial imperfect-recall game (in extensive form) to those in a perfect recall refinement thereof; we refer to the corresponding ratio as the \emph{value of recall (VoR)}. Here, our main contribution is to provide a broad characterization of VoR for different solution concepts in terms of natural game properties.

Many strategic interactions demonstrate that perfect recall offers a significant advantage. In the popular card game blackjack, the house is expected to prevail in the long run against a player with poor recall, but certain memorization strategies tip---at least under the earlier rules followed by casinos---the balance in the player's favor~\citep{Thorp16:Beat}, as pop-culture has hyperbolically portrayed. The role of memory is even more pronounced in other card games such as solitaire~\citep{KIRKPATRICK54:Probability,Foerster13:Solitaire}, where remembering the previously dealt cards drastically increases the odds of winning. We are interested in quantifying how much players benefit from perfect recall.

\subsection{A plot twist: perfect recall can hurt}

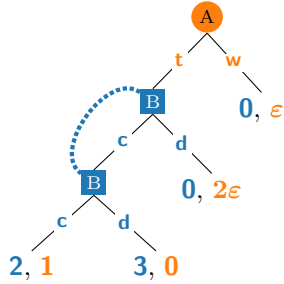
\begin{figure}[t]
    \tikzset{
        every path/.style={-},
        every node/.style={draw},
    }
    \forestset{
    subgame/.style={regular polygon,
    regular polygon sides=3,anchor=north, inner sep=1pt},
    }
    \centering
        \begin{forest}
            [\scriptsize{A},p2gs,name=p20,s sep=25 pt,l sep=21pt
                [\scriptsize{B},p1gs,name=p0,el={2}{t}{},s sep=25pt,l sep=21pt
                    [\scriptsize{B},p1gs,name=p1b,el={1}{c}{},s sep=25pt,l sep=21pt
                        [\util1{2}\text{, }\util2{1},terminal,el={1}{c}{},yshift=-3.3pt]
                        [\util1{3}\text{, }\util2{0},terminal, el={1}{d}{},yshift=-3.3pt]
                    ]
                    [\util1{0}\text{, }\util2{$\boldsymbol{2\varepsilon}$},terminal,el={1}{d}{},yshift=-3.3pt]
                ]
                [\util1{0}\text{, }\util2{$\boldsymbol{\varepsilon} $},terminal,el={2}{w}{},yshift=-3.3pt]
            ]
            \draw[infoset1] (p0) to [bend right=90] (p1b);
        \end{forest}
    \caption{A game with imperfect recall. Giving Bobble (\ponegs) perfect recall hurts both players. Terminals show utilities for Bobble and Alice (\ptwogs). Infosets are joined by dotted lines. 
    \label{fig:recall_bad}}
\end{figure}

The previous examples notwithstanding, surprisingly, endowing a player with perfect recall can end up diminishing every player's utility! Consider \Cref{fig:recall_bad}: Alice has a small amount of money ($\varepsilon>0$) and interacts with an investment bot Bobble, starting from a free trial to see if the bot is defective (\emph{i.e.}, Bobble plays \blued, in which case the game is over and Alice receives a small compensation of $\varepsilon$). If Bobble cooperates (\bluec), the game continues and it gains access to Alice's money, which it multiplies through investments. If Bobble defects (\blued) now, it gets to run away with all the money. However, if it has imperfect recall (cannot remember if the free trial is over), then it has the incentive to cooperate (\bluec) with Alice on both counts, as attempting $\blued$ has a greater chance of causing it to get caught during the free trial. Knowing this, Alice is incentivized to trust (\oranget) Bobble, leading to the cooperative outcome. On the other hand, if Bobble is given perfect recall, it has every incentive to cooperate in the free trial and then defect after getting the money; anticipating this, Alice walks out (\orangew) without interacting with Bobble (\Cref{prop:recall_bad} formalizes this example).

Intuitively, this demonstrates that a player gaining perfect recall can result in the other players trusting it less, eliminating a cooperative outcome that is arbitrarily better for everyone. This is in line with prior work showing that the ability of a player to be simulated by others can benefit everyone in trust-based games \citep{Kovarik23:Game,Kovarik25:Game, Conitzer23:Foundations}.

In spirit, this phenomenon is similar to the famous Braess paradox
~\citep{Braess68:Paradox}, which predicts that augmenting a network with more links can result in worse equilibria. 
We formalize this type of hurtful recall in later sections, and also provide necessary conditions under which it does not arise.

\subsection{Overview of our results}

We formally introduce the value of recall (\Cref{def:vor}) in (imperfect-recall) extensive-form games for a broad set of solution concepts. In particular, building on prior work, our definition is based on the coarsest information refinement of a game that attains perfect recall (\Cref{def:prr}). In the remainder of the paper, we investigate a number of questions relating to the value of recall.

We first formalize the observation made earlier regarding hurtful recall (\Cref{fig:recall_bad}) by showing the existence of games in which a single player getting perfect recall can arbitrarily hurt all players, including themselves, for all the solution concepts considered in this paper (\Cref{prop:recall_bad}). Even more surprisingly, this type of behavior can also arise in single-player games under certain solution concepts (\Cref{ex:refinement_needed}); we argue that this is a pathology as the single player can always choose to ignore information. We show that this issue can be circumvented by replacing each of these solution concepts with an appropriate refinement thereof, one of which is a novel definition (\Cref{defn:edt nash}).

Next, we turn our attention to the computational aspects of the value of recall. We show that VoR is $\NP$-hard to compute, and to approximate, for all solution concepts considered in this paper, even in single-player games (\Cref{thm:hardness}). While this mostly follows from existing hardness results for solving imperfect-recall games~\citep{Tewolde23:Computational}, we prove new hardness results for some solution concepts, which even rule out any multiplicative approximation factor.

Those hardness results notwithstanding, we characterize VoR under optimal play in single-player games based on certain natural properties of the game tree. In particular, we show in \Cref{prop:1p_vor1} that value degradation due to imperfect recall can be fully explained by two sources: \emph{absentmindedness} (an infoset being entered multiple times in a path of play) and external stochasticity. In~\Cref{prop:1p_am,prop:1p_chance}, we provide tight upper bounds for VoR for each of these sources separately. Finally, as our main characterization result, we show that those two bounds compose for games that exhibit both absentmindedness and external stochasticity (\Cref{thm:1p}).

The aforedescribed characterization applies only to optimal play. To extend it to more permissive solution concepts, we make a connection with the price of anarchy literature. Namely, inspired by the homonymous class of games introduced by~\citet{Roughgarden15:Intrinsic}, we introduce the notion of a \emph{smooth} (imperfect-recall) single-player game (\Cref{def:smooth}), and show VoR can be bounded in terms of the smoothness parameters of the game, in conjunction with our previous bound concerning optimal strategies. Besides this connection with the price of anarchy, we further observe that VoR captures some previously studied concepts, such as the \emph{price of uncorrelation} in adversarial team games \citep{Celli18:Computational} and the \emph{price of miscoordination} in security games \citep{Jiang13:Defender}, which enables interpreting their results as bounds on VoR in those games.

Finally, we examine the value of recall with respect to \emph{partial recall}---instead of perfect recall---refinements. In particular, we focus on the natural problem of refining an imperfect-recall game so as to maximize the utility gain, subject to constraining the amount of new recall. We show that, even with oracle access to optimal strategies, that problem is \NP-hard even in single-player games (\Cref{thm:partial}). We conclude with a number of interesting future directions stemming from our work. 

\section{Preliminaries}

Before we proceed, we provide some necessary background on imperfect-recall games and solution concepts for them.

\subsection{Games with imperfect recall}

We start by introducing extensive-form games. We will be following the formalism introduced by~\citet{Fudenberg91:Game_theory}.

\begin{defn}
    An \emph{extensive-form game} $\Gamma$ specifies
    \begin{enumerate}
        \item A rooted tree with node set $\calH$ and edges that represent \emph{actions}. The game starts at the root, and actions are taken to traverse down the tree, until the game finishes at a leaf node, called \emph{terminal node}. The set of terminal nodes is denoted by $\calZ \subset \calH$, and the set of actions available at any nonterminal node $h \in \calH \setminus \calZ$ is denoted by $A_h$.
        \item A finite set $\calN \cup \{c\}$ of $N+1$ players where $N \geq 1$. Set $\calN$ contains the \emph{strategic players}, and $c$ stands for a \emph{chance} ``player'' that models exogenous stochasticity. Each nonterminal node $h$ is assigned to a particular player $i \in \calN \cup \{c\}$, who chooses an action to take from $A_h$. Set $\calH_i$ denotes all nodes assigned to Player $i$.
        \item For each chance node $h \in \calH_c$, a probability distribution $\Prob_c(\cdot \mid h)$ on $A_h$ with which chance elects an action at $h$.
        \item For each strategic player $i \in \calN$, a (without loss of generality) nonnegative \emph{utility (payoff)} function $u_i : \calZ \to \RR_{\geq 0}$ which returns what $i$ receives when the game finishes at a terminal node. Player $i$ aims to maximize that utility.\footnote{Whenever relevant for computational results, we assume all numbers to be rationals represented in binary.}
        \item For each strategic player $i \in \calN$, a partition $\calH_i = \sqcup_{I \in \calI_i} I$ of the nodes of $i$ into information sets (\emph{infosets}). Nodes of the same infoset are considered indistinguishable to the player at that infoset. For that, we also require $A_h = A_{h'}$ for $h, h' \in I$. This also makes action set $A_I$ well-defined.
    \end{enumerate}
\end{defn}

The \emph{game tree} of $\Gamma$ refers to $\calH$, $\{A_h\}_{h\in \calH \setminus \calZ}$, and $\{\Prob_c(\cdot \mid h)\}_{h \in \calH_c}$ (but not its infoset partitioning or utilities). We now formalize games where players may \emph{forget} previously available information.

\begin{defn}[(Im)perfect recall]\label{def:imperf_recall}
For a decision node $h$ of a game $\Gamma$, let $\seq(h) = ( h_k)_{k = 0}^{\depth(h) - 1}$ be the ordered sequence of nodes from the root node $h_0$ to $h$ (excluding $h$) and let $\obs(h)=(i_k, I_k,a_k)_{k=0}^{\depth(h)-1}$ be the corresponding sequence of tuples showing which player $i_k$ acts at $h_k$, the infoset $I_k$ of node $h_k$, and what action $a_k$ was taken at $h_k$ .
For a player $i \in \calN$, let $\obs_i(h)$ be the ordered subsequence of tuples from $\obs(h)$ for which $i_k = i$. We say player $i$ has \emph{perfect recall} in $\Gamma$ if for all of $i$'s infosets $I \in \calI_i$, and all pairs of nodes $h,h' \in I$, we have $\obs_i(h)=\obs_i(h')$. Otherwise, we say Player $i$ has \emph{imperfect recall}. We say that $\Gamma$ is a perfect-recall game if all players $i\in \calN$ have perfect recall in $\Gamma$. Otherwise, we say $\Gamma$ is an imperfect-recall game. 
\end{defn}

\paragraph{Strategies and utilities}

Players can select a probability distribution---a \emph{randomized action}---over the actions at an infoset. A (behavioral) \emph{strategy} $\pi_i$ of a player $i \in \calN$ specifies a randomized action $\pi_i(\cdot \mid I) \in \Delta(A_I)$ at each infoset $I \in \calI_i$. We say $\pi_i$ is \emph{pure} if it assigns probability 1 to a single action for each infoset. A (strategy) \emph{profile} $\pi = (\pi_i)_{i \in \calN}$ specifies a strategy for each player. We use the common notation $\pi_{-i} = (\pi_1, \dots, \pi_{i-1}, \pi_{i+1}, \dots, \pi_n)$. We denote the strategy set of Player $i$ with $S_i$, and $S = \bigtimes_{i \in \calN} S_i$.

We denote the reach probability of a node $h'$ from another node $h$ under a profile $\pi$ as $\Prob(h' \mid \pi, h)$. It evaluates to $0$ if $h \notin \seq(h')$, and otherwise to the product of probabilities with which the actions on the path from $h$ to $h'$ are taken under $\pi$ and chance. We denote with $\U_i(\pi \mid h) \coloneqq \sum_{z \in \calZ} \Prob(z \mid \pi, h) \cdot u_i(z)$ the expected utility of Player $i$ given that the game is at node $h$ and the players are following profile $\pi$. We overload notation for the special case the game starts at root node $h_0$ by defining $\Prob(h \mid \pi) := \Prob(h \mid \pi, h_0)$ and $\U_i(\pi) \coloneqq \U_i(\pi \mid h_0)$.  Finally, let $I^{\text{1st}}$ refer to the nodes $h\in I$ for which $I$ does not appear in $\obs(h)$. Then the reach probability of $I$ (from $h_0$) is $\Prob(I\mid \pi)\coloneqq \sum_{h \in I^{\text{1st}}}\Prob(h \mid \pi)$.

\subsection{Solution concepts}

The value of recall, which we introduce in the next section, does not only depend on the underlying game, but also on our assumptions on what reasoning capabilities each player has. These are formally captured by \emph{solution concepts}.

\paragraph{Nash equilibrium} This is the most classic solution concept in game theory~\citep{Nash50:Non}.

\begin{defn}\label{def:nash}
    A profile $\pi \in S$ is a \emph{\NE{}} of a game $\Gamma$ if for each player $i \in \calN$,
    \begin{align}
        \label{eq:player BR}
        \pi_i \in \argmax_{\pi'_i \in S_i} \U_i(\pi'_i, \pi_{-i}).
    \end{align}
\end{defn}

In the special case that $\Gamma$ is a single-player game, we use the term \emph{optimal strategy} instead of \NE{}.

Unfortunately, a \NE{} is hard to compute, even in a single-player game with imperfect recall \citep{Koller92:Complexity,Gimbert20:Bridge,Tewolde24:Imperfect}. To make matters worse, it may not even exist~\citep{Wichardt08:Existence}. This motivates considering two relaxations based on the \emph{multiselves approach}~\citep{Kuhn53:Extensive}.

\paragraph{Multiselves equilibria} The multiselves approach interprets a player with imperfect recall as a team of multiple instantiations of the player (referred to as \emph{agents} to distinguish from the original player) who independently act at distinct infosets on behalf of the original imperfect-recall player. 

For strategy $\pi_i \in S_i$ of Player $i$, infoset $I \in \calI_i$, and randomized action $\sigma \in \Delta(A_I)$, we denote by $\pi^{I \mapsto \sigma}_i$ the strategy that plays according to $\pi_i$ except at $I$, where it plays $\sigma$.

\begin{defn}
    \label{def:EDT}
    A profile $\pi \in S$ is an \emph{EDT equilibrium} of a game $\Gamma$ if for each player $i \in \calN$ and each of its infosets $I \in \calI_i$, the randomized action $\pi_i( \cdot \mid I)$ satisfies
        $\pi_i( \cdot \mid I) \in \argmax_{\sigma \in \Delta(A_I)} \U_i(\pi^{I \mapsto \sigma}_i, \pi_{-i})$.
\end{defn}

EDT abbreviates \emph{evidential decision theory};
we refer to~\citet{PiccioneR73,Briggs10:Putting,Oesterheld22:Can} for a detailed treatment and motivation. A third equilibrium concept that arose from the aforementioned literature is based on \emph{causal decision theory (CDT)}. It differentiates from EDT only in games with \emph{absentmindedness}, which is when a single infoset $I$ appears multiple times in $\obs(h)$ for some $h\in \calH$ (\Cref{fig:prr}, left). Its original definition is not central to this work and deferred to the appendix. Instead, below we give an equivalent characterization of it \citep{Tewolde24:Imperfect} using \emph{Karush-Kuhn-Tucker (KKT)} points~\citep{Boyd04:Convex}, which generalize the concept of a \emph{stationary point} of a function over an unconstrained domain.

\begin{defn}
    A profile $\pi \in S$ is a \emph{CDT equilibrium} of a game $\Gamma$ if for each player $i \in \calN$, strategy $\pi_i$ is a KKT point of the utility maximization problem (\ref{eq:player BR}).
\end{defn}

These solution concepts form a strict inclusion hierarchy.

\begin{lemma}[\citealp{Oesterheld22:Can}]
\label{lem:EQ hierarchy}
    Nash equilibria are EDT equilibria, which in turn are CDT equilibria.
\end{lemma}
In particular, Nash equilibria are the hardest to compute, but they coincide in the following special cases.

\begin{rem}
\label{rem:edt = cdt w/o abs}
    EDT and CDT equilibria coincide in games without absentmindedness. Nash and EDT equilibria coincide in games with only one infoset per player.
\end{rem}

\section{Value of Recall}

To introduce the novel concept of the value of recall, we first formalize an ordering among infoset partitionings:

\begin{defn}[Game refinements/coarsenings]
    \label{def:ref/coarse}
    Given two extensive-form games $\Gamma$ and $\Gamma'$ with the same game tree and utilities  but potentially different infosets $\{\calI_i\}_{i \in \calN}$ and $\{\calI'_i\}_{i \in \calN}$, and a player $i \in \calN$, we denote $\Gamma \succeq_i \Gamma'$ if for each $I' \in \calI'_i$, there exists $\calJ_i \subseteq\calI_i$ such that $I' = \bigsqcup_{I \in \calJ_i } I$. That is, the infosets in $\calI'$ are (disjointly) partitioned by the infosets in $\calI$. In this case, we say $\Gamma$ (resp. $\Gamma'$) is a \emph{refinement} (\emph{coarsening}) of $\Gamma'$ ($\Gamma$) with respect to player $i$. We denote $\Gamma \succeq \Gamma'$ if $\Gamma \succeq_i \Gamma'$ for all $i \in \calN$ and say $\Gamma$ (resp. $\Gamma'$) is an \emph{all-player} refinement (coarsening) of $\Gamma'$ ($\Gamma$).
\end{defn}

We are now ready to define the perfect recall refinement of an imperfect-recall game.

\begin{defn}[Perfect recall refinements]\label{def:prr}
Given imperfect-recall game $\Gamma$, for all nodes $h \in \calH$ and players $i \in \calN$, define $\obs_i(h)$ as in Definition \ref{def:imperf_recall}. For infoset $I \in \calI_i$ and nodes $h,h' \in I$, define the equivalence relation $h \sim h'$ if $\obs_i(h) = \obs_i(h')$. We say that the \emph{(coarsest) perfect recall refinement} of $\Gamma$ with respect to player $i \in \calN$ is an extensive-form game $\pr_i(\Gamma)$ with the same game tree and utilities as $\Gamma$, but an infoset partition where each $I\in \calI_i$ is partitioned into infosets defined by the equivalence relation $\sim$, and the infosets of all other players are unchanged. The \emph{all-player} (coarsest) perfect recall refinement of $\Gamma$ is an extensive-form game $\pr(\Gamma)$ with the same game tree as $\Gamma$, where the infosets of all players are partitioned as above. 
\end{defn}

An equivalent definition to $\pr(\Gamma)$ was introduced by \citet{Cermak18:Approximating}. Both $\pr_i(\Gamma)$ and $\pr(\Gamma)$ are well-defined, easy to compute, with $\pr_i(\Gamma) \succeq_i \Gamma$ and $\pr(\Gamma)\succeq \Gamma$. As claimed, $\pr_i(\Gamma)$ is the coarsest refinement of $\Gamma$ with respect to $i$ that gives $i$ perfect recall. 
We  formalize this below:

\begin{restatable}{prop}{coarsestrefinement}\label{prop:coarsest}
    Given imperfect-recall game $\Gamma$ and another game $\Gamma'$ that has the same game tree as $\Gamma$ but potentially different infosets, if $\Gamma' \succeq_i \Gamma$ and $i$ has perfect recall in $\Gamma'$, then $\Gamma' \succeq_i \pr_i(\Gamma)$. Moreover, $i$ has perfect recall in $\pr_i(\Gamma)$.\footnote{While intuitive, this last statement is not just definitional: even though nodes $h,h'$ are placed in the same infoset of $\pr_i(\Gamma)$ only if $\obs_i(h)=\obs_i(h')$, the infosets in these sequences are also potentially partitioned, causing the sequences to change too.}
\end{restatable}
(Most proofs are in the appendix due to space constraints.)
\begin{cor}
    If $\Gamma' \succeq \Gamma$ and $\Gamma'$ is a perfect-recall game, then $\Gamma' \succeq \pr(\Gamma)$. Moreover, $\pr(\Gamma)$ is a perfect-recall game.
\end{cor}

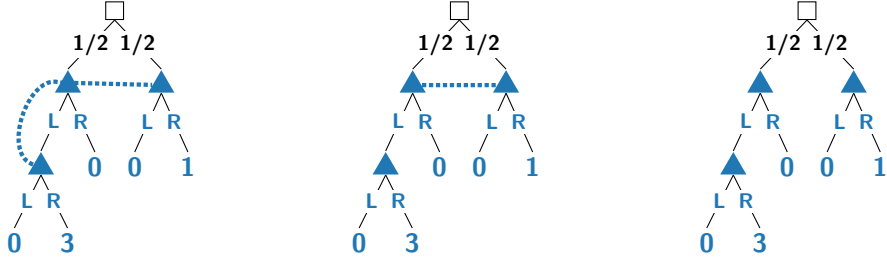
\begin{figure}[t]
\centering
    \tikzset{
        every path/.style={-},
        every node/.style={draw},
    }
    \forestset{
  subgame/.style={regular polygon,
  regular polygon sides=3,anchor=north, inner sep=5pt},
    }
\begin{minipage}{0.1\linewidth}
    \,
\end{minipage}%
\begin{minipage}{.3\linewidth}
     \centering
     \begin{subfigure}{\textwidth}
        
      \begin{forest}
        [,nat,s sep=5pt,el={2}{cont}{}
            [,p1,name=p1a,s sep=10pt,l sep=20pt, el={0}{1/2}{}
                [,p1,name=p1b,s sep=10pt,l sep=20pt, el={1}{L}{},yshift=-2.8pt
                    [\util1{0},terminal,el={1}{L}{},yshift=-6.3pt]
                    [\util1{3},terminal,el={1}{R}{},yshift=-6.3pt]
                ]
                [\util1{0},terminal,el={1}{R}{},yshift=-6.3pt]
            ]
            [,p1,name=p1c,s sep=10pt,l sep=20pt, el={0}{1/2}{}
                [\util1{0},terminal,el={1}{L}{},yshift=-6.3pt]
                [\util1{1},terminal,el={1}{R}{},yshift=-6.3pt]
            ]
        ]   
        \draw[infoset1] (p1b) to [bend left=90] (p1a) -- (p1c);
      \end{forest}
     \label{fig:prr_1}
    \end{subfigure}
    \end{minipage}%
\begin{minipage}{.3\linewidth}
     \centering
     \begin{subfigure}{\textwidth}
        
      \begin{forest}
        [,nat,s sep=5pt,el={2}{cont}{}
            [,p1,name=p1a,s sep=10pt,l sep=20pt, el={0}{1/2}{}
                [,p1,name=p1b,s sep=10pt,l sep=20pt, el={1}{L}{},yshift=-2.8pt
                    [\util1{0},terminal,el={1}{L}{},yshift=-6.3pt]
                    [\util1{3},terminal,el={1}{R}{},yshift=-6.3pt]
                ]
                [\util1{0},terminal,el={1}{R}{},yshift=-6.3pt]
            ]
            [,p1,name=p1c,s sep=10pt,l sep=20pt, el={0}{1/2}{}
                [\util1{0},terminal,el={1}{L}{},yshift=-6.3pt]
                [\util1{1},terminal,el={1}{R}{},yshift=-6.3pt]
            ]
        ]   
        \draw[infoset1] (p1a) -- (p1c);
      \end{forest}
     \label{fig:prr_2}
    \end{subfigure}
    \end{minipage}%
\begin{minipage}{.3\linewidth}
     \centering
     \begin{subfigure}{\textwidth}
        
      \begin{forest}
        [,nat,s sep=5pt,el={2}{cont}{}
            [,p1,name=p1a,s sep=10pt,l sep=20pt, el={0}{1/2}{}
                [,p1,name=p1b,s sep=10pt,l sep=20pt, el={1}{L}{},yshift=-2.8pt
                    [\util1{0},terminal,el={1}{L}{},yshift=-6.3pt]
                    [\util1{3},terminal,el={1}{R}{},yshift=-6.3pt]
                ]
                [\util1{0},terminal,el={1}{R}{},yshift=-6.3pt]
            ]
            [,p1,name=p1c,s sep=10pt,l sep=20pt, el={0}{1/2}{}
                [\util1{0},terminal,el={1}{L}{},yshift=-6.3pt]
                [\util1{1},terminal,el={1}{R}{},yshift=-6.3pt]
            ]
        ]   
      \end{forest}
     \label{fig:prr_3}
    \end{subfigure}
    \end{minipage}%

   \caption{(Left) An imperfect-recall game $\Gamma$. Boxes indicate chance nodes. (Middle) $\pr_1(\Gamma)$, the perfect recall refinement of $\Gamma$ with respect to $\pone$. (Right) $\Gamma$ with perfect information.
   }
   \label{fig:prr}
\end{figure}
By using the coarsest refinement, we seek to isolate the impact of recall on the utility, while filtering out other factors. For instance, the optimal strategy for P1 ($\pone$) in game $\Gamma$ in Figure~\ref{fig:prr}(left) is to play ${\color{p1color}\textbf{\textsf{L}}}$ with probability 1/3, bringing an expected utility of 2/3. If we give the player perfect information, and hence perfect recall in the process, the player can achieve a utility of 2 (\Cref{fig:prr}, right). However, we argue this refinement misrepresents the ``value of recall'' of this game, since P1 now learns the outcome of the chance node, unlike in $\Gamma$. Instead, using the coarsest perfect recall refinement, $\pr(\Gamma)$ per \Cref{def:prr}, leads to utility 3/2 (\Cref{fig:prr}, middle) and properly captures what P1 can gain if its only advantage is to remember everything it once knew. 

The previous example notwithstanding, we should caution that distinguishing perfect recall and perfect information can become blurry: any imperfect information game can be turned into a strategically-equivalent one with only imperfect recall by adding dummy nodes, as we demonstrate in the appendix.

Now, given an imperfect-recall game $\Gamma$, a player of interest (always labelled Player 1), and a solution concept $\SC$, let $u_1(\SC(\Gamma))$  be the utility that Player 1 receives under that solution concept in game $\Gamma$, assuming it exists. In order to ensure that the utility under $\SC$ is uniquely defined (since, for example, there might be multiple Nash equilibria of $\Gamma$ with different utilities for Player 1), we also require $\SC$ to specify whether it is the best or worst possible outcome of that solution concept from Player 1's perspective; this is similar to the definition of solution concepts in the value of commitment~\citep{Letchford14:Value}. In particular, $\bnash, \bedt, \bcdt$ (resp. $\wnash, \wedt, \wcdt$) refer to the best (worst) possible outcome for Player 1 under the corresponding solution concept.

\begin{defn}\label{def:vor}
Given solution concept $\SC$ and $\Gamma$, \emph{the value of recall (VoR) in $\Gamma$ under $\SC$} is
\begin{align*}
\vor^{\SC}(\Gamma)= \frac{u_1(\SC(\pr_1(\Gamma)))}{u_1(\SC(\Gamma))}.
\end{align*}
If we are instead given a game class $\mathscr{C}$, we say that  \emph{the value of recall (VoR) in $\class$ under $\SC$} is
\begin{align*}
\vor^{\SC}(\class)=\sup_{\Gamma \in \class} \frac{u_1(\SC(\pr_1(\Gamma)))}{u_1(\SC(\Gamma))}.
\end{align*}
\end{defn}

We can now formalize the situation that arises in Figure~\ref{fig:recall_bad} and was discussed earlier in the introduction. To do so, we note that strategies $\pi$ and $\pi'$ are \emph{realization-equivalent} if they induce the same reach probability $\Prob(h \mid \pi) = \Prob(h \mid \pi')$ for all $h\in \calH$ (thus achieving the same utility).

\begin{restatable}{prop}{recallbad}
\label{prop:recall_bad}
    For any $\varepsilon>0$, there exists a two-player game $\Gamma$ such that  $\frac{u_i(\SC(\pr_1(\Gamma)))}{u_i(\SC(\Gamma)) } \leq \varepsilon$ for all $i \in \calN$, where $\SC$ is the \emph{only} CDT equilibrium of $\Gamma$, up to realization equivalence. In particular, $\vor^{\SC}(\Gamma)=0$  for  $\SC \in \{\wcdt, 
\bcdt, \wedt, \bedt, \wnash, \bnash\}$.
\end{restatable}

\subsection{Computational complexity of value of recall}

We now show that computing the value of recall is hard. For this theorem alone, we assume (WLOG) for all $z\in \calZ$ that $u_1(z) \geq \eta$ for some $\eta>0$, to ensure VoR is bounded.

\begin{restatable}{thm}{vorhardness}
    \label{thm:hardness}
    Given a game $\Gamma$, computing $\vor^{\SC}(\Gamma)$ is \NP-hard for $\SC \in \{\wcdt, 
\bcdt, \wedt, \bedt,$ $ \wnash, \bnash \}$. Moreover,
    \begin{enumerate}
        \item Unless $\NP = \ZPP$, none of them admits an \FPTAS. In particular, 
        if $\SC \in \{\wcdt, \wedt\}$, then approximation to any multiplicative factor is \NP-hard.
        \item \NP-hardness and conditional inapproximabiltiy hold even if $\Gamma$ is a single-player game. 
    \end{enumerate}
\end{restatable}

A fully polynomial-time approximation scheme (\FPTAS) takes as input a game $\Gamma$, a solution concept $\SC$, and an $\varepsilon>0$ and outputs a number in the interval $(1\pm \varepsilon) \vor^{\SC}(\Gamma)$. Further, $\ZPP$ contains the class of problems solvable by randomized algorithms that always return the correct answer, and whose expected running time is polynomial~\citep{Gill77:Computational}.

Most of the proof of~\Cref{thm:hardness} relies on existing hardness results for equilibrium computation in (single-player) imperfect-recall games~\citep{Koller92:Complexity,Tewolde23:Computational,Gimbert20:Bridge}. The results for $\wcdt$ and $\wedt$ are new, further establishing stronger inapproximability; both proofs proceed by reducing from 3SAT, as we elaborate in the appendix.

\subsection{VoR pathologies and how to fix them}
\label{sec:eqm refine}

While \Cref{prop:recall_bad} shows that getting recall can hurt in general, one would expect this to not be the case in single-player games. Indeed, without any opponents, we would expect giving recall to only benefit the player, since it can always ignore the information it can now recall. This is the case if $\SC$ represents the optimal strategy (\opt), as getting perfect recall expands the strategy set of a player. Further, since the optimal strategy of a game is also its best CDT and EDT equilibrium (\Cref{lem:EQ hierarchy}), we have the following:

\begin{restatable}{prop}{optvorgood}\label{prop:1p_optimal_lowerbound}
    For any single-player game $\Gamma$,
    \begin{align*}
        \vor^{\opt}(\Gamma) =   \vor^{\bedt}(\Gamma) = \vor^{\bcdt}(\Gamma) \geq 1.
    \end{align*}
\end{restatable}
Surprisingly, it turns out that this result in fact does not hold for worst EDT and CDT equilibria of the game:

\begin{ex}\label{ex:refinement_needed}
    Consider the game $\Gamma$ in \Cref{fig:need_refinment_a}. The only CDT/EDT equilibrium of $\Gamma$ is the optimal strategy: always play ${\color{p1color}\textbf{\textsf{L}}}$, bringing a utility of 1. In $\pr_1(\Gamma)$, however, while the same is still the optimal strategy (and hence a CDT and EDT equilibrium), there is now a second CDT and EDT equilibrium: always play ${\color{p1color}\textbf{\textsf{R}}}$ 
on $I_1$ and $I_{21}$, and always play ${\color{p1color}\textbf{\textsf{L}}}$ on $I_{22}$, bringing a utility of $\varepsilon$. Hence, $\vor^{\wedt}(\Gamma)=\vor^{\wcdt}(\Gamma)=\varepsilon$, which can be arbitrarily close to 0.
\end{ex}

\begin{figure} 
    \centering
    \tikzset{
        every path/.style={-},
        every node/.style={draw},
    }
    \forestset{
    subgame/.style={regular polygon,
    regular polygon sides=3,anchor=north, inner sep=1pt},
    }

     \begin{subfigure}{.4\linewidth}
          \centering

        \begin{forest}
            [,p1,name=p0
            [,p1,name=p1a,el={1}{L}{},s sep=15pt
                [\util1{1},el={1}{L}{},terminal]
                [\util1{0},el={1}{R}{},terminal]
            ]
            [,p1,name=p1b,el={1}{R}{},s sep=15pt
                [\util1{$\boldsymbol{\varepsilon}$},el={1}{L}{},terminal]
                [\util1{0},,el={1}{R}{},terminal]
            ]
            ]
            \node[above=0pt of p0,draw=none,p1color]{$I_1$};
            \draw[infoset1] (p1b) to node[below,draw=none,p1color,]{$I_{2}$} (p1a);
        \end{forest}
     \caption{A game $\Gamma$}\label{fig:need_refinment_a}
    \end{subfigure}
      \begin{subfigure}{.55\linewidth}
            \centering
        \begin{forest}
            [,p1,name=p0
            [,p1,name=p1a,el={1}{L}{},s sep=15pt
                [\util1{1},el={1}{L}{},terminal]
                [\util1{0},el={1}{R}{},terminal]
            ]
            [,p1,name=p1b,el={1}{R}{},s sep=15pt
                [\util1{$\boldsymbol{\varepsilon}$},el={1}{L}{},terminal]
                [\util1{0},,el={1}{R}{},terminal]
            ]
            ]
            \node[above=0pt of p0,draw=none,p1color]{$I_1$};
            \node[left=0pt of p1a,draw=none,p1color]{$I_{21}$};
            \node[right=0pt of p1b,draw=none,p1color]{$I_{22}$};
        \end{forest}         \caption{$\pr_1(\Gamma)$}\label{fig:need_refinment_b}
    \end{subfigure}
    \caption{Perfect recall can lead to worse CDT/EDT eq.}
    \label{fig:need_refinment}
\end{figure}
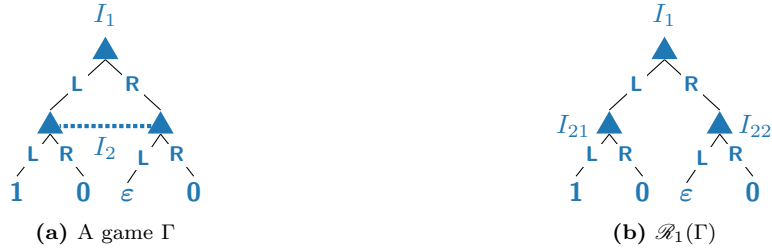

The issue in \Cref{ex:refinement_needed} is that of the chicken or the egg: the unreasonable strategy of playing $\blueR$ at $I_{21}$ cannot violate CDT/EDT conditions if the player never visits $I_{21}$, while if the strategy in $I_{21}$ is unreasonable enough then the decision to not visit $I_{21}$ also does not violate them. This shows that CDT/EDT conditions (which, again, are relaxations of Nash equilibrium) are perhaps \emph{too permissive}, accepting strategies that are not reasonable under perfect recall. To rule out such equilibria, we now introduce \emph{equilibrium refinements} for both solution concepts. (It is important to differentiate between \emph{equilibrium refinements}---which narrows the definition of a solution concept---and \emph{information refinements}, per~\Cref{def:ref/coarse}---which introduces a new game where players have finer infosets.) The refinements of CDT/EDT that we introduce will force the player to consider its behavior in all infosets \emph{it could have} reached, hence preventing pathologies such as \Cref{ex:refinement_needed}.

The appropriate equilibrium refinement for CDT has been introduced by \citet{Lambert19:Equilibria}, which we will refer to as \emph{CDT-Nash}. Due to space constraints, we defer its definition to the appendix. Below, we introduce an analogous, novel refinement called \emph{EDT-Nash}. The relevant properties of both refinements are in \Cref{prop:edt-nash,prop:cdt-nash}.

\begin{defn}
\label{defn:cdt seq rat}
    A strategy $\pi$ in a single-player game $\Gamma$ is \emph{EDT-\rat} if there is a sequence $(\pi^{(k)}, \varepsilon^{(k)})_{k \in \N}$ s.t.
    \begin{enumerate}
        \item each $\pi^{(k)}$ is a strategy in $\Gamma$ such that $\pi^{(k)}(a \mid I) > 0$ for all $I$ and $a$, and $(\pi^{(k)})_{k \in \N}$ converges to $\pi$;
        \item each $\varepsilon^{(k)} > 0$ and $(\varepsilon^{(k)})_{k \in \N}$  converge to $0$; and
        \item for each $k$, for all $I$ with $\Prob(I \mid \pi^{(k)}) >0 $ and $\sigma \in \Delta(A_I)$, 
        \begin{align*}
            \frac{1}{\Prob(I \mid \pi^{(k)})} \cdot \Big( \U_1(\pi^{(k), I \mapsto \sigma}) - \U_1(\pi^{(k)}) \Big) \leq \varepsilon^{(k)}.
        \end{align*}
    \end{enumerate}
\end{defn}
Intuitively, the sequence of fully mixed strategies prevents the player from ignoring any infosets it could have reached.

\begin{defn}
\label{defn:edt nash}
    A profile $\pi$ is an \emph{EDT-\NE{}} of $\Gamma$ if it is an EDT equilibrium and if for all $i \in \calN$, and in the single-player perspective of $\Gamma$ (where every other player plays fixed $\pi_{-i}$), the strategy $\pi_i$ is realization-equivalent to an EDT-\rat{} strategy $\pi$. 
\end{defn}
The key property of our refinement is that it agrees with the optimal strategy under perfect recall.
\begin{restatable}{prop}{edtnash}
    \label{prop:edt-nash}
    EDT-\NEs{} are EDT equilibria. Without absentmindedness, a strategy profile is an EDT-\NE{} iff it is a CDT-\NE{}. Under perfect recall, a strategy profile is an EDT-\NE{} iff it is a \NE{}.
\end{restatable}
An analogous result was shown by \citet{Lambert19:Equilibria} for CDT-Nash equilibria:
\begin{prop}[\citealp{Lambert19:Equilibria}]\label{prop:cdt-nash}
    CDT-\NEs{} are CDT equilibria, and they always exist. Under perfect recall, a strategy profile is a CDT-\NE{} iff it is a \NE{} of $\Gamma$.
\end{prop}

The above propositions imply that in a single-player game with perfect recall, the only CDT-Nash and EDT-Nash equilibrium is the optimal strategy. Combined with \Cref{prop:1p_optimal_lowerbound}, this shows that the refinements successfully resolve the pathologies that arose with CDT/EDT.
\begin{cor}
    \label{cor:path-fixed}
    Given single-player game $\Gamma$, $\vor^{\SC}(\Gamma) \geq 1$ for $\SC \in \{\wcdtnash, \bcdtnash,$ $\wedtnash, \bedtnash\}$.
\end{cor}

\section{Bounding the Value of Recall}

In this section, we first focus on bounding $\vor^{\opt}$ for single-player games, and show that while it can be arbitrarily large in general, we can still parameterize it using properties of the game tree and the utility functions. Later on, we show how $\vor^{\SC}$ for other solution concepts can be bounded in conjunction with these parametrizations.

A key observation is that in single-player games, there are exactly two factors that can lead to a change in the optimal utility when perfect recall is introduced: (1) absentmindedness, and (2) chance nodes. Indeed, if neither is present, the optimal utility remains unchanged.

\begin{restatable}{prop}{optone}\label{prop:1p_vor1}
    For a single-player game $\Gamma$ with no chance nodes and with no absentmindedness, $\vor^{\opt}(\Gamma)=1$. Further, for both $\Gamma$ and $\pr_1(\Gamma)$, there is a \emph{pure} optimal strategy.
\end{restatable}

As we will show, either absentmindedness or chance nodes is sufficient to have a game with $\vor^{\opt}(\Gamma)>1$. We first deal with each of these cases separately, before moving on to games that exhibit both.

\paragraph{VoR due to absentmindedness}

To bound the impact of absentmindedness, we first parameterize the number of times an infoset is visited and an action is taken on the way to a leaf node. Given a single-player game $\Gamma$, for any $z\in \calZ$ with $\obs(z)=(i_k, I_k,a_k)_{k=0}^{\depth(z)-1}$, for each $I \in \calI_1$ and $a \in A_I$ let $n_z(I)=|\{k: I_k=I\}|$, $n_z(a)=|\{k: a_k=a\}|$, and $p_z(a) = \frac{n_z(a)}{n_z(I)}$. Then, we define
\begin{align*} 
    \am(z) = \prod_{\substack{{I \in \calI_1: n_z(I)>1}  \\{a \in A_I: n_z(a)>0}}} p_z(a)^{n_z(a)} \quad \in (0,1]
\end{align*}
to be the \emph{absentmindedness coefficient} of $z$. Intuitively, it describes how easy it is to reach $z$ under absentmindedness:

\begin{restatable}{lemma}{eminne}\label{lemma:am_strategy}
    Given single-player game $\Gamma$ with no chance nodes, for all $z \in \calZ$, there exists a strategy $\pi_z$ that reaches $z$ with probability $\am(z)$, achieving $u_1(\pi_z) \geq \am(z)u_1(z)$. 
\end{restatable}

We are now ready to introduce our upper bounds for $\vor$ in terms of the absentmindedness coefficients:  
\begin{restatable}{prop}{ambound}\label{prop:1p_am}
    In a single-player game $\Gamma$ without chance nodes, we have
    \begin{align*}
        \vor^{\opt}(\Gamma) &\leq \frac{\max_{z\in \calZ}u_1(z)}{\max_{z \in \calZ} \am(z)u_1(z)}  \leq \frac{1}{\am(z^*)}
    \end{align*}
    where $z^* = \argmax_{z \in \calZ} u_1(z)$.
\end{restatable}

As we see next, the inequalities in \Cref{prop:1p_am} are tight.

\begin{figure} 
    \tikzset{
        every path/.style={-},
        every node/.style={draw},
    }
    \forestset{
    subgame/.style={regular polygon,
    regular polygon sides=3,anchor=center, inner sep=1pt},
    }
    \begin{minipage}{.15\linewidth}
    \,
    \end{minipage}%
    \begin{minipage}{.35\linewidth}
     \centering
     \begin{subfigure}{\textwidth}

    \begin{forest}
        [,p1,name=p1a,s sep=25pt,l sep=21pt
            [,p1,name=p1b,el={1}{L}{},s sep=25pt,l sep=21pt
                [,p1,name=p1c,el={1}{L}{},s sep=25pt,l sep=21pt
                    [\util1{0},terminal,el={1}{L}{},yshift=-3.3pt]
                    [,p1,name=p1d,el={1}{R}{},s sep=25pt,l sep=21pt
                        [\util1{0},terminal, el={1}{L}{},yshift=-3.3pt]
                        [\util1{1},terminal, el={1}{R}{},yshift=-3.3pt]
                    ]
                ]
                [\util1{0},terminal,el={1}{R}{},yshift=-3.3pt]
            ]     
            [\util1{0},terminal,el={1}{R}{},yshift=-3.3pt]
        ]
        \draw[infoset1] (p1a) to [bend right=90] (p1b) to[bend right=90] (p1c) to[bend left=90] (p1d);
    \end{forest}
    \end{subfigure}
    \end{minipage}%
    \begin{minipage}{.35\linewidth}
      \centering
      \begin{subfigure}{\textwidth}
        \begin{forest}
        [,nat,s sep=5pt,el={2}{cont}{}
            [,p1,name=p1a,s sep=10pt,l sep=20pt, el={0}{a}{}
                [,p1,name=p1b,s sep=10pt,l sep=20pt, el={1}{a}{},yshift=-2.8pt
                    [\util1{1},terminal,el={1}{a}{},yshift=-6.3pt]
                    [\util1{0},terminal,el={1}{b}{},yshift=-6.3pt]
                ]
                [,p1,name=p1c,s sep=10pt,l sep=20pt, el={1}{b}{},yshift=-2.8pt
                    [\util1{0},terminal,el={1}{a}{},yshift=-6.3pt]
                    [\util1{0},terminal,el={1}{b}{},yshift=-6.3pt]
                ]
            ]
            [,p1,name=p1d,s sep=10pt,l sep=20pt, el={0}{b}{}
                [,p1,name=p1e,s sep=10pt,l sep=20pt, el={1}{a}{},yshift=-2.8pt
                    [\util1{0},terminal,el={1}{a}{},yshift=-6.3pt]
                    [\util1{0},terminal,el={1}{b}{},yshift=-6.3pt]
                ]
                [,p1,name=p1f,s sep=10pt,l sep=20pt, el={1}{b}{},yshift=-2.8pt
                    [\util1{0},terminal,el={1}{a}{},yshift=-6.3pt]
                    [\util1{1},terminal,el={1}{b}{},yshift=-6.3pt]
                ]
            ]
        ]   
        \draw[infoset1] (p1b) -- (p1c) -- (p1e) -- (p1f);
      \end{forest}
    \end{subfigure}
    \end{minipage}%
    \caption{(Left) \Cref{ex:1am_tight}, $n=4$. (Right) Ex.~\ref{ex:1chance_tight}, $n=2$}
    \label{fig:tight games}
\end{figure}

\begin{ex}\label{ex:1am_tight}
    Consider a single-player game $\Gamma$ where Lenny needs to pick between the action $\blueL$ and the action $\blueR$ for $n$ consecutive rounds for some even $n$. He gets utility 1 if he first plays $\blueL$ exactly $n/2$ times followed by $\blueR$ the remaining $n/2$ times. If he does anything else, the game is over and he gets 0 utility. Moreover, his memory is reset each time.
    
    The game tree of $\Gamma$ has $n$ nodes, $n+1$ leaves, and a single infoset $I$. \Cref{fig:tight games}(Left) depicts this for $n=4$. Let $z^*$ be the single leaf node with $u_1(z^*)=1$. The optimal strategy in $\pr_1(\Gamma)$ (where each node is its own infoset) is to arrive at $z^*$, guaranteeing a utility of 1. In $\Gamma$, however, Lenny cannot do anything better than playing uniformly at random, achieving an expected utility of $2^{-n}$. Moreover, $\am(z^*)=2^{-n}$. Hence, for $\Gamma$, all of the inequalities in \Cref{prop:1p_am} are tight.
 \end{ex}

\Cref{ex:1am_tight} is the worst-case scenario with regard to absentmindedness: only one leaf node brings positive utility, and reaching it requires playing each action equally often.

Importantly, $\am(z)$ is independent of the utilities of $\Gamma$. This allows us to interpret \Cref{prop:1p_am} in two parts: a tighter bound on $\Gamma$ using its utilities, and another bound that applies to all games that differ from $\Gamma$ only by their utility functions.

\begin{restatable}{cor}{amclass}\label{cor:1pm_am}
Given a single-player game $\Gamma$ without chance nodes, say $\class$ is the class of games that share the same game tree and infoset partition as $\Gamma$. Then
\begin{align*}
    \vor^{\opt}(\class) = \max_{z \in \calZ} \frac{1}{\am(z)}.
\end{align*}
\end{restatable}

\paragraph{VoR due to chance nodes} We now do a similar analysis for chance nodes. Given a single-player game $\Gamma$, for any $z\in \calZ$ with $\obs(z)=(i_k, I_k,a_k)_{k=0}^{\depth(z)-1}$, say $k_1,k_2, \ldots k_\ell$ are steps that correspond to chance nodes, \emph{i.e.}, $i_{k_j} =c$ for all $j \in [\ell]$. Then, the \emph{chance coefficient} of leaf node $z$ is
\begin{align*} 
    \chance(z) = \prod_{j=1}^\ell \Prob_c(a_{k_j} | h_{k_j})
\end{align*}
if $\ell>0$  and $\chance(z)=1$ otherwise. $\chance (z)$ is the probability of reaching $z$ in $\pr_1(\Gamma)$ (\emph{i.e.}, under perfect recall), given that the player is trying to reach it. For each chance node $h \in \calH_c$, and each $a\in A_h$, say $H_{ha} \subset \calH_c$ are the chance nodes in the subtree rooted at $ha$ (the node reached when chance plays $a$ at $h$). Then the \emph{branching factor} of $h$ is $\branch(h)=\sum_{a \in A_h}b_h(a)$, where
\begin{align*}
b_h(a) = \begin{cases} 1 & \text{if }|H_{ha}|=0\\ \max_{h \in H_{ha}}\beta(h)& \text{otherwise}
\end{cases}.
\end{align*}
One can compute $\branch(h)$ for each $h\in \calH_c$ recursively, starting from the bottom of the tree (that is, the leaf nodes). We now have all the tools we need for characterizing the impact of chance nodes on VoR.
\begin{restatable}{prop}{chancebound}\label{prop:1p_chance}
    In a single-player game $\Gamma$ without absentmindedness, we have\footnote{By convention, we assume $\max_{h \in \calH_c} \beta(h)=1$ if $\calH_c =\emptyset.$}
    \begin{align*}
        \vor^{\opt}(\Gamma) \leq \frac{u_1(\opt(\pr_1(\Gamma)))}{\max_{z \in \calZ} \chance(z)u_1(z)}  \leq \max_{h \in \calH_c} \beta(h).
    \end{align*}
\end{restatable}
\begin{cor}
Given a single-player game $\Gamma$ without absentmindedness, say $\class$ is the class of games that share the same game tree and infoset partition as $\Gamma$. Then
\begin{align*}
    \vor^{\opt}(\class)  \leq \max_{h \in \calH_c}\beta(h).
\end{align*}
\end{cor}
The reason we have an inequality for the game class, unlike in \Cref{cor:1pm_am}, is that while absentmindedness does imply imperfect recall, chance nodes alone do not tell us anything about the information structure of the game. We now show that the bounds in \Cref{prop:1p_chance} are also tight.

\begin{ex}\label{ex:1chance_tight}
    Consider a game $\Gamma$ that starts with a single chance node $h_c$ with $|A_{h_c}|=n$, each played with equal probability. Under each outcome, Dory needs to act twice, using the same action set as chance $A_{h_c}$, and gets utility 1 only if she replicates the action of the chance node both times, and 0 otherwise. Each of Dory's nodes immediately following the chance node is in its own information set of size 1, and every other node is in a single information set. \Cref{fig:tight games}(Right) shows the game tree for $n=2$. 

    In $\pr_1(\Gamma)$, Dory has perfect information and can guarantee utility 1. However, with imperfect recall, the best she can do is select the correct action the first time she acts, and then any strategy she will follow on the large information set will bring her expected utility $1/n$. Moreover, $\beta(h_c)=n$ and $\max_{z\in \calZ} \chi(z)u_1(z)=1/n$, showing that for $\Gamma$ all the inequalities in \Cref{prop:1p_chance} are tight.
\end{ex}

We end this section by showing that our results from Propositions \ref{prop:1p_am} and \ref{prop:1p_chance} do in fact compose, hence giving a parameterization of $\vor^\opt$ for any single-player game.

\begin{restatable}{thm}{onethm}\label{thm:1p}
    For a single-player game $\Gamma$, say $\class$ is the class of games that share the same game tree and infosets. Then
    \begin{align*}
       \vor^\opt(\class) \leq \max_{z \in \calZ, h \in \calH_c} \frac{\branch(h)}{\am(z)}.
    \end{align*}
\end{restatable}
\subsection{Smooth imperfect-recall games}
\label{sec:smoothness}

Remaining on single-player games, here we bound the value of recall for a broader set of equilibria. Our approach is driven by a connection with the \emph{price of anarchy (PoA)}. In particular, we introduce the notion of a \emph{smooth} (single-player) imperfect-recall game, which is based on the homonymous class of (multi-player) games by~\citet{Roughgarden15:Intrinsic}. Below, we denote by $(\pi_I)_{ I \in \calI_1} \in S$ the player's strategy, and use the notation $\pi_{-I} \coloneqq (\pi_{I'})_{I' \in \calI_1 \setminus \{ I \}}$.

\begin{defn}
    \label{def:smooth}
    A single-player game $\Gamma$ is \emph{$(\lambda, \mu)$-smooth} if there exists $\pistar \in S$ such that for any $\pi \in S$,
    \begin{equation}
        \label{eq:smooth}
        \frac{1}{|\calI_1|} \sum_{ I \in \calI_1} u_1(\pistar_I, \pi_{-I}) \geq \lambda u_1(\opt(\Gamma)) - \mu u_1(\pi).
    \end{equation}
\end{defn}

The rationale behind this definition is that it enables disentangling the left-hand side of~\eqref{eq:smooth} via a suitable strategy $\pistar$, with the property that if followed by each infoset separately, a non-trivial fraction of the optimal utility can be secured \emph{no matter the strategy in the rest of the infosets}. While this may appear like an overly restrictive property, it manifests itself in many important applications~\citep{Roughgarden17:Price}. In \Cref{def:smooth}, infosets play the role of strategic players in Roughgarden's formalism; we provide a concrete example of a smooth imperfect-recall game in the appendix.

Now, by definition, an EDT equilibrium $\pi$ satisfies $u_1 ( \pi ) \geq \frac{1}{|\calI_1|} \sum_{ I \in \calI_1} u_1 (\pistar_I, \pi_{-I})$ (by applying~\Cref{def:EDT} successively for each infoset). Combining with~\eqref{eq:smooth}, we immediately arrive at the following conclusion.

\begin{prop}
    \label{prop:smooth}
    Let $\Gamma$ be a $(\lambda, \mu)$-smooth, single-player game. For any EDT equilibrium $\pi \in S$,
    \[
        u_1(\pi) \geq \frac{\lambda}{1 + \mu} u_1(\opt(\Gamma)).
    \]
\end{prop}

In words, $\rho \coloneqq \lambda/(1 + \mu)$ measures the degradation incurred in an EDT equilibrium, which is referred to as the \emph{robust price of anarchy} in the parlance of~\citet{Roughgarden15:Intrinsic}. In light of~\Cref{prop:smooth}, bounding $\vor^{\wedt}(\Gamma)$ reduces to relating $u_1(\opt(\Gamma))$ in terms of $u_1(\opt(\pr(\Gamma)))$, which was accomplished earlier in~\Cref{thm:1p}.

\begin{cor}
    \label{cor:smoothness}
    Let $\Gamma$ be a $(\lambda, \mu)$-smooth, single-player game. Then
    \begin{equation*}
        \vor^{\wedt} (\Gamma) \leq \frac{1 + \mu}{\lambda} \max_{ z \in \calZ, h \in \calH_c} \frac{\beta(h)}{\am(z)}.
    \end{equation*}
\end{cor}

This also applies to $\wedtnash$ always and (by~\Cref{rem:edt = cdt w/o abs}) to $\wcdt$ and $\wcdtnash$ when $\Gamma$ has no absentmindedness.

\subsection{Further connections} In addition, we note that the value of recall (\Cref{def:vor}) encompasses several notions from prior literature. First, the \emph{price of uncorrelation} in adversarial team games \citep{Celli18:Computational}, which measures how much of a team of players facing a single adversary can gain from communicating, corresponds to $\vor^{\bnash}(\class^{2p0s})$, where $\class^{2p0s}$ is the class of two-player zero-sum games (based on their construction, one of the players---corresponding to the adversary---has perfect recall).

Second, the \emph{price of miscoordination} in security games \citep{Jiang13:Defender}, which measures the utility loss due to having multiple defenders rather than a single one, corresponds to the VoR in this game class; here, $\SC$ corresponds to \emph{Stackelberg equilibria}, which involves Player 1 committing to a strategy and its opponent best responding.  We expand on the above connections in the appendix.
\section{Partial Recall Refinements}

So far, we have defined the value of recall based on the (coarsest) perfect-recall refinement. It is also natural to consider the change in utility due to obtaining \emph{partial} recall.

\begin{defn}[Partial recall refinements]\label{def:partial_rr}
 Given games $\Gamma$ and $\Gamma'$ with the same game tree but possibly different info sets, $\Gamma'$ is a \emph{partial recall refinement} of $\Gamma$ with respect to a player $i \in \calN$ if $\pr_i(\Gamma) \succeq_i \Gamma' \succeq_i \Gamma$. Further, $\Gamma'$ is an \emph{all-player} partial recall refinement of $\Gamma$ if $\pr(\Gamma) \succeq \Gamma' \succeq \Gamma$.
\end{defn}

\begin{figure} 
    \centering
    \tikzset{
        every path/.style={-},
        every node/.style={draw},
    }
    \forestset{
    subgame/.style={regular polygon,
    regular polygon sides=3,anchor=north, inner sep=1pt},
    }

     \begin{subfigure}{.4\linewidth}
          \centering

        \begin{forest}
            [,p1,name=p0
            [,p1,name=p1a,el={1}{L}{},s sep=15pt
                [\util1{2},el={1}{L}{},terminal]
                [\util1{0},el={1}{R}{},terminal]
            ]
            [,p1,name=p1b,el={1}{R}{},s sep=15pt
                [\util1{0},el={1}{L}{},terminal]
                [\util1{1},,el={1}{R}{},terminal]
            ]
            ]
            \draw[infoset1] (p1b)--(p1a) to [bend left=70] (p0);
        \end{forest}
     \caption{A game $\Gamma$}
    \end{subfigure}
      \begin{subfigure}{.55\linewidth}
            \centering
        \begin{forest}
            [,p1,name=p0
            [,p1,name=p1a,el={1}{L}{},s sep=15pt
                [\util1{2},el={1}{L}{},terminal]
                [\util1{0},el={1}{R}{},terminal]
            ]
            [,p1,name=p1b,el={1}{R}{},s sep=15pt
                [\util1{0},el={1}{L}{},terminal]
                [\util1{1},,el={1}{R}{},terminal]
            ]
            ]
            \draw[infoset1] (p1b)--(p1a);
        \end{forest}
          \caption{A partial recall refinement of $\Gamma$}
    \end{subfigure}
    \caption{Partial recall gives a worse EDT-Nash equilibrium: In (a), the only EDT-Nash equilibrium is playing $\blueL$; in (b), playing $\blueR$ in both infosets is also an EDT-Nash equilibrium.}
    \label{fig:partial_recall}
\end{figure}
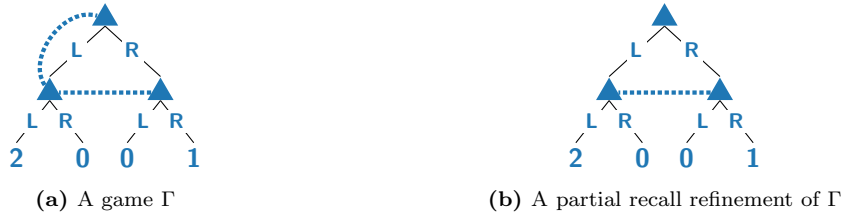

Partial recall refinements introduces further interesting properties. For example, \Cref{fig:partial_recall} shows that partial recall can lead to a worse EDT-Nash equilibrium in a single-player game; this stands in contrast to perfect recall refinements (\Cref{cor:path-fixed}).

In what follows, we study the complexity of perhaps the most natural problem arising from~\Cref{def:partial_rr}: how should one refine an imperfect-recall game so as to maximize the utility gain, subject to constraining the amount of new recall. This problem is well-motivated from the literature on abstraction (\emph{e.g.},~\citealt{Kroer16:Imperfect}), but to our knowledge, it has not been studied in this form. To formalize constraints on recall, we first introduce the following notation: consider games $\Gamma, \Gamma'$ that differ solely on infosets for $i \in \calN$ ($\calI_i$ and $\calI'_i$, respectively). We write $\Gamma' \vdash_i \Gamma$ if there is $I \in \calI_i$ such that $\calI_i'= \calI_i \setminus \{I\} \cup \{I_1,I_2\}$, where $I =I_1 \sqcup I_2$; \emph{i.e.}, $\Gamma'$ results from splitting a single infoset of $i$ in $\Gamma$.

\begin{defn}
   Fix a player $i\in \calN$. We say $\Gamma$ is its own \emph{$0$-partial recall refinement}. $ \Gamma'$ is a \emph{$k$-partial recall refinement} of $\Gamma$ if it is a partial recall refinement of $\Gamma$ and $\Gamma' \vdash_i \Gamma''$, where $\Gamma''$ is some $(k-1)$-partial recall refinement of $\Gamma$.
\end{defn}

This restriction is motivated by the fact that many practical algorithms scale with the number of infosets, and so one naturally strives to minimize that when abstracting a game~\citep{Kroer14:Extensive,Kroer16:Imperfect}.

Then, the computational problem $k$-$\bestpartial(\Gamma)$ asks: given a parameter $k \in \mathbb{N}$ and a single-player game $\Gamma$, compute its $k$-partial recall refinement $\Gamma'$ that maximizes $u_1(\opt(\Gamma'))$. For this task, we assume to be given access to an oracle $\mathcal{O}$ that outputs the optimal utility of any single-player game; even though such an oracle can only make the problem easier, we show the following hardness result.

\begin{restatable}{thm}{partialhard}
    \label{thm:partial}
     $k$-$\bestpartial(\Gamma)$ is \NP-hard.
\end{restatable}

Our proof relies on a reduction from \emph{exact cover by 3-sets} \cite{Garey79:Computers}, which asks to exactly cover a set of items using a given family of subsets of size three. Our construction consists of a chance node with an action per item, followed by player nodes with an action per subset.

\section{Further Related Work}

Starting from the seminal work of~\citet{Koller92:Complexity}, there has been much interest in characterizing the computational complexity of natural equilibrium concepts in single- and multi-player imperfect-recall games~\citep{Tewolde23:Computational,Tewolde24:Imperfect,Gimbert20:Bridge,Lambert19:Equilibria}; this undertaking is driven by several key applications, many of which were discussed in the prequel. In particular, for many problems of interest, perfect recall is known to be necessary for efficient computation, although that is no longer the case for more structured classes of games~\citep{Lanctot12:No}. Such hardness results pertaining to imperfect-recall games translate to intractability for computing the value of recall in our context (\Cref{thm:hardness}).

\paragraph{Game abstraction} The main theme of our work, which revolves around quantifying the value of recall, bears resemblance to certain considerations in the line of work on abstraction; as we explained earlier, imperfect-recall abstractions arise naturally when attempting to compress the description of the game by way of merging information sets. This is very much related to ``computational games,'' in which agents are charged for computation~\citep{Halpern13:Sequential,Sandholm00:Issues,Larson01:Bargaining,Larson01:Costly}; in such settings, choosing to forget information may indeed be rational. For the theoretical foundation of abstraction techniques based on imperfect recall, we refer to~\citet{Kroer18:Unified,Kroer16:Imperfect,Kroer14:Extensive}. Those papers examine the problem of computing different solution concepts in imperfect-recall games, and then mapping back to the perfect-recall refinement. In particular, \citet{Kroer16:Imperfect} consider a certain class of imperfect-recall games, which they refer to as \emph{chance-relaxed skew well-formed (CRSWF)} games---extending an earlier work by~\citet{Lanctot12:No}. This is a particular, somewhat benign class of imperfect-recall games in which there exists an information-set refinement (akin to~\Cref{def:ref/coarse}) that satisfies certain natural properties. We should mention that~\Cref{thm:partial}, which concerns identifying a partial recall refinement maximizing the utility gain, relates to a hardness result of~\citet[Theorem 4.1]{Kroer16:Arbitrage-Free} for computing a CRSWF abstraction (starting from a perfect-recall game) with minimal error bounds; their proof proceeds by a reduction from clustering in the plane. A class of games related to CRSWF is defined through \emph{A-loss recall}, which requires that each loss of a player's memory can be traced back to loss of memory of his own action~\citep{Kaneko95:Behavior,Kline02:Minimum,Cermak18:Approximating}. Instead, here we treat general extensive-form games with imperfect recall.

\paragraph{Related concepts to VoR} Quantifying the value of recall is conceptually related to the \emph{price of anarchy}~\citep{Roughgarden15:Intrinsic,Koutsoupias99:Worst,Roughgarden00:How}, another influential line of work in algorithmic game theory. The price of anarchy measures the welfare loss caused by players confined to a worst-case (Nash) equilibrium compared to the welfare-optimal state---the related notion of the \emph{price of stability}~\citep{Anshelevich08:Price} instead deals with the best-case equilibrium. Our work makes a connection with that line of work by leveraging the so-called \emph{smoothness} framework of~\citet{Roughgarden15:Intrinsic}.

Similar but distinct notions studied in the literature include the \emph{price of malice}~\citep{Babaioff09:Congestion,Moscibroda06:When}, which incorporates in price of anarchy a small number of Byzantine agents; the \emph{mediation value}~\citep{Ashlagi08:Value}, which quantifies the improvement in welfare when expanding the set of outcomes from Nash equilibria to the more permissive set of correlated equilibria; the \emph{defender miscoordination} in security games~\citep{Jiang13:Defender}, which captures the loss incurred by miscoordination on the defenders' part; the closely related \emph{price of uncorrelation} in adversarial team games~\citep{Celli18:Computational}; the \emph{value of commitment}~\citep{Letchford14:Value}, which measures a player's benefit derived from the power to commit; and the \emph{value of information}, a concept with a long history in economics---known to be potentially hurtful in the presence of multiple players~\citep{Bassan03:Positive,Hirshleifer71:Private}---\emph{cf.}~\Cref{prop:recall_bad}. The connection between the value of recall and some of the above concepts was made explicit in our paper.

Those similarities notwithstanding, the value of recall presents several differences compared to the mediation value (MV) and the price of anarchy (PoA). Both of those latter concepts study the change in \emph{social welfare} (the sum over all players' utilities), with MV enriching the set of outcomes from Nash to correlated equilibria, and PoA positing a benevolent central authority that imposes the welfare-maximizing state. In our setting, on the other hand, Player 1 is selfishly trying to maximize its own own utility (just like all other players), and VoR measures that player's benefit---or, indeed, the cost (\Cref{prop:recall_bad})---under perfect recall. In this sense, VoR is perhaps more similar to the value of commitment (VoC)~\citep{Letchford14:Value}, which quantifies the impact of another strategic device---namely, commitments---again on the utility of a selfish player. Furthermore, \Cref{prop:recall_bad} marks another distinction of VoR from PoA and MV, as the latter two must always be at least 1 (by definition). VoC must also be at least 1 if \emph{mixed commitments} are allowed, but can be smaller than 1 if the player is constrained to \emph{pure commitments}---for instance, in rock-paper-scissors.

Before we proceed, it is also worth connecting the value of recall to \emph{situational awareness}, a concept that has engendered a significant body of work~\citep{Endsley00:Theoretical,Stanton01:Siutational}. In particular, recall constitutes an important component shaping the situational awareness of an agent, albeit not the only one.

\paragraph{Nonmonotonicities} We highlighted earlier the counterintuitive fact that perfect recall can hurt players' utilities (\Cref{fig:recall_bad}). This type of nonmonotonicity is not without precedent; besides Braess paradox, which was cited earlier, we should mention two other similar phenomena: \citet{Waugh09:Abstractiona} showed that a more fine-grained abstraction of a game can result in worse equilibria for zero-sum games, while~\citet{Jagadeesan23:Improved} observed that improving Bayes risk can decrease the overall predictive accuracy across users for a marketplace consisting of competing model-providers.

\paragraph{POMDPs and repeated games} Finally, while our main focus here is on extensive-form games, the role of memory in policy optimization is also central in the context of partially observable Markov decision processes (POMDPs)~\citep{Astrom65:Optimal,Kaelbling98:Planning,Bonet09}. This undertaking often manifests itself in characterizing the gap between non-Markovian and Markovian policies; for example, we refer to~\citet{Mutti22:Importance}, and references therein. Relatedly, that discrepancy has been the subject of much work in the theory of repeated games in economics~\citep{Barlo09:Repeated,Aumann89:Cooperation,Bhaskar02:Asynchronous,Cole05:Finite,Foster18:Smooth,Papadimitriou94:Complexity}. An interesting direction for future work is whether such results can be cast in the language of VoR introduced in our paper.

\section{Conclusions and Future Research}

We introduced the value of recall, which measures the utility gain by granting a player perfect recall. Our work opens many interesting avenues for future research. First, the value of recall could be used to guide abstraction techniques. We also observed the interesting phenomenon that perfect recall can be hurtful to all players. It would be interesting to provide a broader characterization of games where this is so---a natural candidate being \emph{simulation games}~\citep{Kovarik24:Recursive}, and quantify the \emph{price} of recall therein. Furthermore, we have focused on the value of recall from the perspective of a single player, but understanding the impact on \emph{social welfare} is a natural next step.

\section*{Acknowledgments}

We are grateful to the anonymous AAAI reviewers for many helpful comments that improved the exposition of this paper. We also thank Brian Hu Zhang for many discussions. Ratip Emin Berker, Emanuel Tewolde, and Vincent Conitzer thank the Cooperative AI Foundation, Polaris Ventures (formerly the Center for Emerging Risk Research) and Jaan Tallinn’s donor-advised fund at Founders Pledge for financial support. Tuomas Sandholm is supported by the Vannevar Bush Faculty Fellowship ONR N00014-23-1-2876, National Science Foundation grants RI-2312342 and RI-1901403, ARO award W911NF2210266, and NIH award A240108S001.

\bibliography{dairefs}

\begin{thebibliography}{89}
\providecommand{\natexlab}[1]{#1}
\providecommand{\url}[1]{\texttt{#1}}
\expandafter\ifx\csname urlstyle\endcsname\relax
  \providecommand{\doi}[1]{doi: #1}\else
  \providecommand{\doi}{doi: \begingroup \urlstyle{rm}\Url}\fi

\bibitem[Anshelevich et~al.(2008)Anshelevich, Dasgupta, Kleinberg, Tardos,
  Wexler, and Roughgarden]{Anshelevich08:Price}
Elliot Anshelevich, Anirban Dasgupta, Jon~M. Kleinberg, {\'{E}}va Tardos, Tom
  Wexler, and Tim Roughgarden.
\newblock The price of stability for network design with fair cost allocation.
\newblock \emph{SIAM Journal on Computing}, 38\penalty0 (4):\penalty0
  1602--1623, 2008.

\bibitem[Ashlagi et~al.(2008)Ashlagi, Monderer, and
  Tennenholtz]{Ashlagi08:Value}
Itai Ashlagi, Dov Monderer, and Moshe Tennenholtz.
\newblock On the value of correlation.
\newblock \emph{Journal of Artificial Intelligence Research}, 33:\penalty0
  575--613, 2008.

\bibitem[Aumann and Sorin(1989)]{Aumann89:Cooperation}
Robert~J Aumann and Sylvain Sorin.
\newblock Cooperation and bounded recall.
\newblock \emph{Games and Economic Behavior}, 1\penalty0 (1):\penalty0 5--39,
  1989.

\bibitem[Babaioff et~al.(2009)Babaioff, Kleinberg, and
  Papadimitriou]{Babaioff09:Congestion}
Moshe Babaioff, Robert Kleinberg, and Christos~H. Papadimitriou.
\newblock Congestion games with malicious players.
\newblock \emph{Games and Economic Behavior}, 67\penalty0 (1):\penalty0 22--35,
  2009.

\bibitem[Bakhtin et~al.(2022)Bakhtin, Brown, Dinan, Farina, Flaherty, Fried,
  Goff, Gray, Hu, Jacob, Komeili, Konath, Kwon, Lerer, Lewis, Miller, Mitts,
  Renduchintala, Roller, Rowe, Shi, Spisak, Wei, Wu, Zhang, and
  Zijlstra]{Bakhtin22:Human}
Anton Bakhtin, Noam Brown, Emily Dinan, Gabriele Farina, Colin Flaherty, Daniel
  Fried, Andrew Goff, Jonathan Gray, Hengyuan Hu, Athul~Paul Jacob, Mojtaba
  Komeili, Karthik Konath, Minae Kwon, Adam Lerer, Mike Lewis, Alexander~H.
  Miller, Sasha Mitts, Adithya Renduchintala, Stephen Roller, Dirk Rowe, Weiyan
  Shi, Joe Spisak, Alexander Wei, David Wu, Hugh Zhang, and Markus Zijlstra.
\newblock Human-level play in the game of diplomacy by combining language
  models with strategic reasoning.
\newblock \emph{Science}, 378\penalty0 (6624):\penalty0 1067--1074, 2022.

\bibitem[Barlo et~al.(2009)Barlo, Carmona, and Sabourian]{Barlo09:Repeated}
Mehmet Barlo, Guilherme Carmona, and Hamid Sabourian.
\newblock Repeated games with one-memory.
\newblock \emph{Journal of Economic Theory}, 144\penalty0 (1):\penalty0
  312--336, 2009.

\bibitem[Basilico et~al.(2017)Basilico, Celli, Nittis, and
  Gatti]{Basilico17:Team}
Nicola Basilico, Andrea Celli, Giuseppe~De Nittis, and Nicola Gatti.
\newblock Team-maxmin equilibrium: Efficiency bounds and algorithms.
\newblock In \emph{Conference on Artificial Intelligence (AAAI)}, 2017.

\bibitem[Bassan et~al.(2003)Bassan, Gossner, Scarsini, and
  Zamir]{Bassan03:Positive}
Bruno Bassan, Olivier Gossner, Marco Scarsini, and Shmuel Zamir.
\newblock Positive value of information in games.
\newblock \emph{International Journal of Game Theory}, 32\penalty0
  (1):\penalty0 17--31, 2003.

\bibitem[Bhaskar and Vega-Redondo(2002)]{Bhaskar02:Asynchronous}
Venkataraman Bhaskar and Fernando Vega-Redondo.
\newblock Asynchronous choice and markov equilibria.
\newblock \emph{Journal of Economic Theory}, 103\penalty0 (2):\penalty0
  334--350, 2002.

\bibitem[Bonet(2009)]{Bonet09}
Blai Bonet.
\newblock Deterministic pomdps revisited.
\newblock In \emph{Proceedings of the Conference on Uncertainty in Artificial
  Intelligence (UAI)}, pages 59--66, 2009.

\bibitem[Bowling et~al.(2015)Bowling, Burch, Johanson, and
  Tammelin]{Bowling15:Heads}
Michael Bowling, Neil Burch, Michael Johanson, and Oskari Tammelin.
\newblock Heads-up limit hold'em poker is solved.
\newblock \emph{Science}, 347\penalty0 (6218):\penalty0 145--149, 2015.

\bibitem[Boyd and Vandenberghe(2004)]{Boyd04:Convex}
Stephen Boyd and Lieven Vandenberghe.
\newblock \emph{Convex Optimization}.
\newblock Cambridge University Press, 2004.

\bibitem[Braess(1968)]{Braess68:Paradox}
Dietrich Braess.
\newblock {\"U}ber ein paradoxon aus der {V}erkehrsplanung.
\newblock \emph{Unternehmensforschung}, 12:\penalty0 258--268, 1968.

\bibitem[Briggs(2010)]{Briggs10:Putting}
Rachael Briggs.
\newblock Putting a value on beauty.
\newblock In \emph{Oxford Studies in Epistemology: Volume 3}, pages 3--34.
  Oxford University Press, 2010.

\bibitem[Brown and Sandholm(2018)]{Brown17:Superhuman}
Noam Brown and Tuomas Sandholm.
\newblock Superhuman {AI} for heads-up no-limit poker: {Libratus} beats top
  professionals.
\newblock \emph{Science}, 359\penalty0 (6374):\penalty0 418--424, 2018.

\bibitem[Brown et~al.(2015)Brown, Ganzfried, and
  Sandholm]{Brown15:Hierarchical}
Noam Brown, Sam Ganzfried, and Tuomas Sandholm.
\newblock Hierarchical abstraction, distributed equilibrium computation, and
  post-processing, with application to a champion no-limit {T}exas {H}old'em
  agent.
\newblock In \emph{International Conference on Autonomous Agents and
  Multi-Agent Systems (AAMAS)}, 2015.

\bibitem[Celli and Gatti(2018)]{Celli18:Computational}
Andrea Celli and Nicola Gatti.
\newblock Computational results for extensive-form adversarial team games.
\newblock In \emph{Conference on Artificial Intelligence (AAAI)}, 2018.

\bibitem[{\v{C}}erm{\'a}k et~al.(2017){\v{C}}erm{\'a}k, Bosansk{\'{y}}, and
  Lis{\'{y}}]{Cermak17:Algorithm}
Ji{\v{r}}{\'\i} {\v{C}}erm{\'a}k, Branislav Bosansk{\'{y}}, and Viliam
  Lis{\'{y}}.
\newblock An algorithm for constructing and solving imperfect recall
  abstractions of large extensive-form games.
\newblock In \emph{Proceedings of the International Joint Conference on
  Artificial Intelligence (IJCAI)}, 2017.

\bibitem[{\v{C}}erm{\'a}k et~al.(2018){\v{C}}erm{\'a}k, Bo{\v{s}}ansk{\`y},
  Hor{\'a}k, Lis{\`y}, and P{\v{e}}chou{\v{c}}ek]{Cermak18:Approximating}
Ji{\v{r}}{\'\i} {\v{C}}erm{\'a}k, Branislav Bo{\v{s}}ansk{\`y}, Karel
  Hor{\'a}k, Viliam Lis{\`y}, and Michal P{\v{e}}chou{\v{c}}ek.
\newblock Approximating maxmin strategies in imperfect recall games using
  a-loss recall property.
\newblock \emph{International Journal of Approximate Reasoning}, 93:\penalty0
  290--326, 2018.

\bibitem[Cesa-Bianchi and Lugosi(2006)]{Cesa-Bianchi06:Predictiona}
Nicolo Cesa-Bianchi and G{\'a}bor Lugosi.
\newblock \emph{Prediction, learning, and games}.
\newblock Cambridge university press, 2006.

\bibitem[Chen et~al.(2024)Chen, Ghersengorin, and Petersen]{Chen2024:Imperfect}
Eric~O. Chen, Alexis Ghersengorin, and Sami Petersen.
\newblock Imperfect recall and {AI} delegation, 2024.

\bibitem[Cole and Kocherlakota(2005)]{Cole05:Finite}
Harold~L. Cole and Narayana~R. Kocherlakota.
\newblock Finite memory and imperfect monitoring.
\newblock \emph{Games and Economic Behavior}, 53\penalty0 (1):\penalty0 59--72,
  2005.

\bibitem[Conitzer(2019)]{Conitzer19:Designing}
Vincent Conitzer.
\newblock Designing preferences, beliefs, and identities for artificial
  intelligence.
\newblock In \emph{Conference on Artificial Intelligence (AAAI)}, 2019.

\bibitem[Conitzer and Oesterheld(2023)]{Conitzer23:Foundations}
Vincent Conitzer and Caspar Oesterheld.
\newblock Foundations of cooperative {AI}.
\newblock In \emph{Conference on Artificial Intelligence (AAAI)}, 2023.

\bibitem[Emmons et~al.(2022)Emmons, Oesterheld, Critch, Conitzer, and
  Russell]{Emmons22:Learning}
Scott Emmons, Caspar Oesterheld, Andrew Critch, Vincent Conitzer, and Stuart
  Russell.
\newblock For learning in symmetric teams, local optima are global nash
  equilibria.
\newblock In \emph{International Conference on Machine Learning (ICML)}, 2022.

\bibitem[Endsley and Garland(2000)]{Endsley00:Theoretical}
Mica~R Endsley and Daniel~J Garland.
\newblock Theoretical underpinnings of situation awareness: A critical review.
\newblock \emph{Situation awareness analysis and measurement}, 1\penalty0
  (1):\penalty0 3--21, 2000.

\bibitem[Foerster and Wattenhofer(2013)]{Foerster13:Solitaire}
Klaus-Tycho Foerster and Roger Wattenhofer.
\newblock The solitaire memory game.
\newblock Technical report, ETH Zurich, 2013.

\bibitem[Foster and Hart(2018)]{Foster18:Smooth}
Dean~P. Foster and Sergiu Hart.
\newblock Smooth calibration, leaky forecasts, finite recall, and nash
  dynamics.
\newblock \emph{Games and Economic Behavior}, 109:\penalty0 271--293, 2018.

\bibitem[Fudenberg and Tirole(1991)]{Fudenberg91:Game_theory}
Drew Fudenberg and Jean Tirole.
\newblock \emph{Game Theory}.
\newblock MIT Press, October 1991.

\bibitem[Ganzfried and Sandholm(2014)]{Ganzfried14:Potential}
Sam Ganzfried and Tuomas Sandholm.
\newblock Potential-aware imperfect-recall abstraction with earth mover's
  distance in imperfect-information games.
\newblock In \emph{Conference on Artificial Intelligence (AAAI)}, 2014.

\bibitem[Garey and Johnson(1979)]{Garey79:Computers}
Michael Garey and David Johnson.
\newblock \emph{Computers and Intractability}.
\newblock W. H. Freeman and Company, 1979.

\bibitem[Gill(1977)]{Gill77:Computational}
John Gill.
\newblock Computational complexity of probabilistic turing machines.
\newblock \emph{SIAM Journal on Computing}, 6\penalty0 (4):\penalty0 675--695,
  1977.

\bibitem[Gimbert et~al.(2020)Gimbert, Paul, and Srivathsan]{Gimbert20:Bridge}
Hugo Gimbert, Soumyajit Paul, and B.~Srivathsan.
\newblock A bridge between polynomial optimization and games with imperfect
  recall.
\newblock In \emph{Autonomous Agents and Multi-Agent Systems}, 2020.

\bibitem[Goemans et~al.(2004)Goemans, Li, Mirrokni, and
  Thottan]{Goemans04:Market}
Michel~X. Goemans, Li~(Erran) Li, Vahab~S. Mirrokni, and Marina Thottan.
\newblock Market sharing games applied to content distribution in ad-hoc
  networks.
\newblock In \emph{Interational Symposium on Mobile Ad Hoc Networking and
  Computing}, 2004.

\bibitem[Halpern and Pass(2013)]{Halpern13:Sequential}
Joseph~Y. Halpern and Rafael Pass.
\newblock Sequential equilibrium in computational games.
\newblock In \emph{Proceedings of the Twenty-Third international joint
  conference on Artificial Intelligence}, pages 171--176. AAAI Press, 2013.

\bibitem[Hirshleifer(1971)]{Hirshleifer71:Private}
Jack Hirshleifer.
\newblock The private and social value of information and the reward to
  inventive activity.
\newblock \emph{The American Economic Review}, 61\penalty0 (4):\penalty0
  561--574, 1971.

\bibitem[Jagadeesan et~al.(2023)Jagadeesan, Jordan, Steinhardt, and
  Haghtalab]{Jagadeesan23:Improved}
Meena Jagadeesan, Michael~I. Jordan, Jacob Steinhardt, and Nika Haghtalab.
\newblock Improved bayes risk can yield reduced social welfare under
  competition.
\newblock In \emph{Proceedings of the Annual Conference on Neural Information
  Processing Systems (NeurIPS)}, 2023.

\bibitem[Jiang and Leyton-Brown(2011)]{Jiang11:Polynomial}
Albert Jiang and Kevin Leyton-Brown.
\newblock Polynomial-time computation of exact correlated equilibrium in
  compact games.
\newblock In \emph{Proceedings of the ACM Conference on Electronic Commerce
  (EC)}, 2011.

\bibitem[Jiang et~al.(2013)Jiang, Procaccia, Qian, Shah, and
  Tambe]{Jiang13:Defender}
Albert~Xin Jiang, Ariel~D. Procaccia, Yundi Qian, Nisarg Shah, and Milind
  Tambe.
\newblock Defender (mis)coordination in security games.
\newblock In \emph{Proceedings of the International Joint Conference on
  Artificial Intelligence (IJCAI)}, 2013.

\bibitem[Johanson et~al.(2013)Johanson, Burch, Valenzano, and
  Bowling]{Johanson13:Evaluating}
Michael Johanson, Neil Burch, Richard Valenzano, and Michael Bowling.
\newblock Evaluating state-space abstractions in extensive-form games.
\newblock In \emph{International Conference on Autonomous Agents and
  Multi-Agent Systems (AAMAS)}, 2013.

\bibitem[Kaelbling et~al.(1998)Kaelbling, Littman, and
  Cassandra]{Kaelbling98:Planning}
Leslie~Pack Kaelbling, Michael~L. Littman, and Anthony~R. Cassandra.
\newblock Planning and acting in partially observable stochastic domains.
\newblock \emph{Artificial Intelligence}, 101\penalty0 (1-2):\penalty0 99--134,
  1998.

\bibitem[Kaneko and Kline(1995)]{Kaneko95:Behavior}
Mamoru Kaneko and J~Jude Kline.
\newblock Behavior strategies, mixed strategies and perfect recall.
\newblock \emph{International Journal of Game Theory}, 24:\penalty0 127--145,
  1995.

\bibitem[Kirkpatrick(1954)]{KIRKPATRICK54:Probability}
Paul Kirkpatrick.
\newblock Probability theory of a simple card game.
\newblock \emph{The Mathematics Teacher}, 47\penalty0 (4):\penalty0 245--248,
  1954.

\bibitem[Kline(2002)]{Kline02:Minimum}
J~Jude Kline.
\newblock Minimum memory for equivalence between ex ante optimality and
  time-consistency.
\newblock \emph{Games and Economic Behavior}, 38\penalty0 (2):\penalty0
  278--305, 2002.

\bibitem[Koller and Megiddo(1992)]{Koller92:Complexity}
Daphne Koller and Nimrod Megiddo.
\newblock The complexity of two-person zero-sum games in extensive form.
\newblock \emph{Games and Economic Behavior}, 4\penalty0 (4):\penalty0
  528--552, October 1992.

\bibitem[Korzhyk et~al.(2011)Korzhyk, Yin, Kiekintveld, Conitzer, and
  Tambe]{Korzhyk11::Stackelberg}
Dmytro Korzhyk, Zhengyu Yin, Christopher Kiekintveld, Vincent Conitzer, and
  Milind Tambe.
\newblock Stackelberg vs. nash in security games: An extended investigation of
  interchangeability, equivalence, and uniqueness.
\newblock \emph{Journal of Artificial Intelligence Research}, 41:\penalty0
  297--327, 2011.

\bibitem[Koutsoupias and Papadimitriou(1999)]{Koutsoupias99:Worst}
Elias Koutsoupias and Christos Papadimitriou.
\newblock Worst-case equilibria.
\newblock In \emph{Symposium on Theoretical Aspects in Computer Science}, 1999.

\bibitem[Kovar{\'{\i}}k et~al.(2023)Kovar{\'{\i}}k, Oesterheld, and
  Conitzer]{Kovarik23:Game}
Vojtech Kovar{\'{\i}}k, Caspar Oesterheld, and Vincent Conitzer.
\newblock Game theory with simulation of other players.
\newblock In \emph{Proceedings of the International Joint Conference on
  Artificial Intelligence (IJCAI)}, 2023.

\bibitem[Kovar{\'{\i}}k et~al.(2024)Kovar{\'{\i}}k, Oesterheld, and
  Conitzer]{Kovarik24:Recursive}
Vojtech Kovar{\'{\i}}k, Caspar Oesterheld, and Vincent Conitzer.
\newblock Recursive joint simulation in games, 2024.

\bibitem[Kovar{\'{\i}}k et~al.(2025)Kovar{\'{\i}}k, Sauerberg, Hammond, and
  Conitzer]{Kovarik25:Game}
Vojtech Kovar{\'{\i}}k, Nathaniel Sauerberg, Lewis Hammond, and Vincent
  Conitzer.
\newblock Game theory with simulation in the presence of unpredictable
  randomisation.
\newblock In \emph{International Conference on Autonomous Agents and
  Multi-Agent Systems (AAMAS)}, 2025.

\bibitem[Kreps and Wilson(1982)]{Kreps82:Sequential}
David~M. Kreps and Robert Wilson.
\newblock Sequential equilibria.
\newblock \emph{Econometrica}, 50\penalty0 (4):\penalty0 863--894, 1982.

\bibitem[Kroer and Sandholm(2014)]{Kroer14:Extensive}
Christian Kroer and Tuomas Sandholm.
\newblock Extensive-form game abstraction with bounds.
\newblock In \emph{Proceedings of the ACM Conference on Economics and
  Computation (EC)}, 2014.

\bibitem[Kroer and Sandholm(2016)]{Kroer16:Imperfect}
Christian Kroer and Tuomas Sandholm.
\newblock Imperfect-recall abstractions with bounds in games.
\newblock In \emph{Proceedings of the ACM Conference on Economics and
  Computation (EC)}, 2016.

\bibitem[Kroer and Sandholm(2018)]{Kroer18:Unified}
Christian Kroer and Tuomas Sandholm.
\newblock A unified framework for extensive-form game abstraction with bounds.
\newblock In \emph{Proceedings of the Annual Conference on Neural Information
  Processing Systems (NeurIPS)}, 2018.

\bibitem[Kroer et~al.(2016)Kroer, Dud{\'{\i}}k, Lahaie, and
  Balakrishnan]{Kroer16:Arbitrage-Free}
Christian Kroer, Miroslav Dud{\'{\i}}k, S{\'{e}}bastien Lahaie, and Sivaraman
  Balakrishnan.
\newblock Arbitrage-free combinatorial market making via integer programming.
\newblock In \emph{Proceedings of the 2016 {ACM} Conference on Economics and
  Computation, {EC} '16, Maastricht, The Netherlands, July 24-28, 2016}, pages
  161--178, 2016.

\bibitem[Kuhn(1953)]{Kuhn53:Extensive}
H.~W. Kuhn.
\newblock Extensive games and the problem of information.
\newblock In \emph{Contributions to the Theory of Games}, volume~2 of
  \emph{Annals of Mathematics Studies, 28}, pages 193--216. Princeton
  University Press, 1953.

\bibitem[Lambert et~al.(2019)Lambert, Marple, and Shoham]{Lambert19:Equilibria}
Nicolas~S Lambert, Adrian Marple, and Yoav Shoham.
\newblock On equilibria in games with imperfect recall.
\newblock \emph{Games and Economic Behavior}, 113:\penalty0 164--185, 2019.

\bibitem[Lanctot et~al.(2012)Lanctot, Gibson, Burch, Zinkevich, and
  Bowling]{Lanctot12:No}
Marc Lanctot, Richard Gibson, Neil Burch, Martin Zinkevich, and Michael
  Bowling.
\newblock No-regret learning in extensive-form games with imperfect recall.
\newblock In \emph{International Conference on Machine Learning (ICML)}, 2012.

\bibitem[Larson and Sandholm(2001{\natexlab{a}})]{Larson01:Bargaining}
Kate Larson and Tuomas Sandholm.
\newblock Bargaining with limited computation: Deliberation equilibrium.
\newblock \emph{Artificial Intelligence}, 132\penalty0 (2):\penalty0 183--217,
  2001{\natexlab{a}}.

\bibitem[Larson and Sandholm(2001{\natexlab{b}})]{Larson01:Costly}
Kate Larson and Tuomas Sandholm.
\newblock Costly valuation computation in auctions.
\newblock In \emph{Theoretical Aspects of Rationality and Knowledge (TARK
  VIII)}, 2001{\natexlab{b}}.

\bibitem[Letchford et~al.(2014)Letchford, Korzhyk, and
  Conitzer]{Letchford14:Value}
Joshua Letchford, Dmytro Korzhyk, and Vincent Conitzer.
\newblock On the value of commitment.
\newblock \emph{Autonomous Agents and Multi-Agent Systems}, 28:\penalty0
  986--1016, 2014.

\bibitem[Meisheri and Khadilkar(2020)]{Meisheri20:Sample}
Hardik Meisheri and Harshad Khadilkar.
\newblock Sample efficient training in multi-agent adversarial games with
  limited teammate communication, 2020.

\bibitem[Morav{\v c}{\'i}k et~al.(2017)Morav{\v c}{\'i}k, Schmid, Burch,
  Lis{\'y}, Morrill, Bard, Davis, Waugh, Johanson, and
  Bowling]{Moravvcik17:DeepStack}
Matej Morav{\v c}{\'i}k, Martin Schmid, Neil Burch, Viliam Lis{\'y}, Dustin
  Morrill, Nolan Bard, Trevor Davis, Kevin Waugh, Michael Johanson, and Michael
  Bowling.
\newblock Deepstack: Expert-level artificial intelligence in heads-up no-limit
  poker.
\newblock \emph{Science}, 356\penalty0 (6337):\penalty0 508--513, 2017.

\bibitem[Moscibroda et~al.(2006)Moscibroda, Schmid, and
  Wattenhofer]{Moscibroda06:When}
Thomas Moscibroda, Stefan Schmid, and Roger Wattenhofer.
\newblock When selfish meets evil: byzantine players in a virus inoculation
  game.
\newblock In \emph{Proceedings of the ACM Symposium on Principles of
  Distributed Computing}, 2006.

\bibitem[Mutti et~al.(2022)Mutti, Santi, and Restelli]{Mutti22:Importance}
Mirco Mutti, Riccardo~De Santi, and Marcello Restelli.
\newblock The importance of non-markovianity in maximum state entropy
  exploration.
\newblock In \emph{International Conference on Machine Learning (ICML)}, 2022.

\bibitem[Nash(1950)]{Nash50:Non}
John Nash.
\newblock \emph{Non-cooperative games}.
\newblock PhD thesis, Priceton University, 1950.

\bibitem[Oesterheld and Conitzer(2024)]{Oesterheld22:Can}
Caspar Oesterheld and Vincent Conitzer.
\newblock Can {\em de se} choice be {\em ex ante} reasonable in games of
  imperfect recall? a complete analysis, 2024.

\bibitem[Papadimitriou and Roughgarden(2008)]{Papadimitriou08:Computing}
Christos~H. Papadimitriou and Tim Roughgarden.
\newblock Computing correlated equilibria in multi-player games.
\newblock \emph{Journal of the ACM}, 55\penalty0 (3):\penalty0 14:1--14:29,
  2008.

\bibitem[Papadimitriou and Yannakakis(1994)]{Papadimitriou94:Complexity}
Christos~H. Papadimitriou and Mihalis Yannakakis.
\newblock On complexity as bounded rationality (extended abstract).
\newblock In \emph{Proceedings of the Annual Symposium on Theory of Computing
  (STOC)}, 1994.

\bibitem[Piccione and Rubinstein(1997)]{PiccioneR73}
Michele Piccione and Ariel Rubinstein.
\newblock On the interpretation of decision problems with imperfect recall.
\newblock \emph{Games and Economic Behavior}, 20:\penalty0 3--24, 1997.

\bibitem[Resnick et~al.(2020)Resnick, Gao, M{\'a}rton, Osogami, Pang, and
  Takahashi]{Resnick2020:Pommerman}
Cinjon Resnick, Chao Gao, G{\"o}r{\"o}g M{\'a}rton, Takayuki Osogami, Liang
  Pang, and Toshihiro Takahashi.
\newblock Pommerman {\&} {NeurIPS} 2018.
\newblock In \emph{The NeurIPS '18 Competition}, pages 11--36. Springer
  International Publishing, 2020.

\bibitem[Roughgarden(2015)]{Roughgarden15:Intrinsic}
Tim Roughgarden.
\newblock Intrinsic robustness of the price of anarchy.
\newblock \emph{Journal of the ACM}, 62\penalty0 (5):\penalty0 32:1--32:42,
  2015.

\bibitem[Roughgarden and Tardos(2000)]{Roughgarden00:How}
Tim Roughgarden and {\'{E}}va Tardos.
\newblock How bad is selfish routing?
\newblock In \emph{Proceedings of the Annual Symposium on Foundations of
  Computer Science (FOCS)}, November 2000.

\bibitem[Roughgarden et~al.(2017)Roughgarden, Syrgkanis, and
  Tardos]{Roughgarden17:Price}
Tim Roughgarden, Vasilis Syrgkanis, and {\'{E}}va Tardos.
\newblock The price of anarchy in auctions.
\newblock \emph{Journal of Artificial Intelligence Research}, 59:\penalty0
  59--101, 2017.

\bibitem[Sandholm(2000)]{Sandholm00:Issues}
Tuomas Sandholm.
\newblock Issues in computational {V}ickrey auctions.
\newblock \emph{International Journal of Electronic Commerce}, 4\penalty0
  (3):\penalty0 107--129, 2000.
\newblock Early version in ICMAS-96.

\bibitem[Stanton et~al.(2001)Stanton, Chambers, and
  Piggott]{Stanton01:Siutational}
N.A Stanton, P.R.G Chambers, and J~Piggott.
\newblock Situational awareness and safety.
\newblock \emph{Safety Science}, 39\penalty0 (3):\penalty0 189--204, 2001.

\bibitem[Tewolde et~al.(2023)Tewolde, Oesterheld, Conitzer, and
  Goldberg]{Tewolde23:Computational}
Emanuel Tewolde, Caspar Oesterheld, Vincent Conitzer, and Paul~W. Goldberg.
\newblock The computational complexity of single-player imperfect-recall games.
\newblock In \emph{Proceedings of the International Joint Conference on
  Artificial Intelligence (IJCAI)}, 2023.

\bibitem[Tewolde et~al.(2024)Tewolde, Zhang, Oesterheld, Zampetakis, Sandholm,
  Goldberg, and Conitzer]{Tewolde24:Imperfect}
Emanuel Tewolde, Brian~Hu Zhang, Caspar Oesterheld, Manolis Zampetakis, Tuomas
  Sandholm, Paul~W. Goldberg, and Vincent Conitzer.
\newblock Imperfect-recall games: Equilibrium concepts and their complexity.
\newblock In \emph{Proceedings of the International Joint Conference on
  Artificial Intelligence (IJCAI)}, 2024.

\bibitem[Thorp(2016)]{Thorp16:Beat}
Edward~O Thorp.
\newblock \emph{Beat the dealer: A winning strategy for the game of
  twenty-one}.
\newblock Vintage, 2016.

\bibitem[Vetta(2002)]{Vetta02:Nash}
Adrian Vetta.
\newblock Nash equilibria in competitive societies, with applications to
  facility location, traffic routing and auctions.
\newblock In \emph{Proceedings of the Annual Symposium on Foundations of
  Computer Science (FOCS)}, 2002.

\bibitem[{von Stengel}(1996)]{Stengel96:Efficient}
Bernhard {von Stengel}.
\newblock Efficient computation of behavior strategies.
\newblock \emph{Games and Economic Behavior}, 14\penalty0 (2):\penalty0
  220--246, 1996.

\bibitem[{von Stengel} and Koller(1997)]{VonStengel97:Team}
Bernhard {von Stengel} and Daphne Koller.
\newblock Team-maxmin equilibria.
\newblock \emph{Games and Economic Behavior}, 21\penalty0 (1):\penalty0
  309--321, 1997.

\bibitem[Waugh et~al.(2009{\natexlab{a}})Waugh, Schnizlein, Bowling, and
  Szafron]{Waugh09:Abstractiona}
Kevin Waugh, David Schnizlein, Michael Bowling, and Duane Szafron.
\newblock Abstraction pathologies in extensive games.
\newblock In \emph{International Conference on Autonomous Agents and
  Multi-Agent Systems (AAMAS)}, 2009{\natexlab{a}}.

\bibitem[Waugh et~al.(2009{\natexlab{b}})Waugh, Zinkevich, Johanson, Kan,
  Schnizlein, and Bowling]{Waugh09:Practical}
Kevin Waugh, Martin Zinkevich, Michael Johanson, Morgan Kan, David Schnizlein,
  and Michael Bowling.
\newblock A practical use of imperfect recall.
\newblock In \emph{Symposium on Abstraction, Reformulation and Approximation
  (SARA)}, 2009{\natexlab{b}}.

\bibitem[Wichardt(2008)]{Wichardt08:Existence}
Philipp~C. Wichardt.
\newblock Existence of nash equilibria in finite extensive form games with
  imperfect recall: A counterexample.
\newblock \emph{Games and Economic Behavior}, 63\penalty0 (1):\penalty0
  366--369, 2008.

\bibitem[Zhang and Sandholm(2022)]{Zhang22:Polynomial}
Brian~Hu Zhang and Tuomas Sandholm.
\newblock Polynomial-time optimal equilibria with a mediator in extensive-form
  games.
\newblock In \emph{Proceedings of the Annual Conference on Neural Information
  Processing Systems (NeurIPS)}, 2022.

\bibitem[Zhang et~al.(2023)Zhang, Farina, and Sandholm]{Zhang23:Team}
Brian~Hu Zhang, Gabriele Farina, and Tuomas Sandholm.
\newblock Team belief {DAG:} generalizing the sequence form to team games for
  fast computation of correlated team max-min equilibria via regret
  minimization.
\newblock In \emph{International Conference on Machine Learning (ICML)}, 2023.

\bibitem[Zhang et~al.(2022)Zhang, An, and Subrahmanian]{Zhang22:Correlation}
Youzhi Zhang, Bo~An, and V.~S. Subrahmanian.
\newblock Correlation-based algorithm for team-maxmin equilibrium in
  multiplayer extensive-form games.
\newblock In \emph{Proceedings of the International Joint Conference on
  Artificial Intelligence (IJCAI)}, 2022.

\bibitem[Åström({1965})]{Astrom65:Optimal}
Karl~Johan Åström.
\newblock {Optimal Control of Markov Processes with Incomplete State
  Information I}.
\newblock \emph{{Journal of Mathematical Analysis and Applications}},
  {10}:\penalty0 {174--205}, {1965}.

\end{thebibliography}

\appendix
\section{Omitted Proofs}

We dedicate this section to the proofs omitted from the main body.

\subsection{Proof of \Cref{prop:coarsest}}
We first prove that $\pr_i(\Gamma)$ in fact corresponds the \emph{coarsest} perfect recall refinement of $\Gamma$ from player $i$'s perspective.
\coarsestrefinement*
\begin{proof}
    We first prove that $i$ has perfect recall in $\pr_i(\Gamma)$. Say that $\calI_i$ and $\calI_i^*$ are the infosets of player $i$ in $\Gamma$ and $\pr_i(\Gamma)$, respectively. For any node $h \in \calH$, say $\obs_i(h)$ and $\obs^*_i(h)$ are as defined in \Cref{def:imperf_recall} for $\Gamma$ and $\pr_i(\Gamma)$, respectively.
    Take any $I \in \calI^*_i$ and $h^{(1)},h^{(2)} \in I$ with $\seq_i(h^{(1)})= (h^{(1)}_k)_{k=1}^{\ell^{(1)}}$ and $\seq_i(h^{(2)})= (h^{(2)}_k)_{k=1}^{\ell^{(2)}}$ ($\seq_i$ is the same for $\Gamma$ and $\pr_i(\Gamma)$ as they share the same game tree). We would like to show that $\obs_i^*(h^{(1)})=\obs_i^*(h^{(2)})$. By the infoset partition defined in \Cref{def:prr}, $h^{(1)},h^{(2)} \in I$ implies $\obs_i(h^{(1)})=\obs_i(h^{(2)}) \equiv ( i, I_k, a_k)_{k=1}^\ell$; thus, we must have $\ell^{(1)}=\ell^{(2)}=\ell$. Fix any $k \in [\ell]$ and consider $h^{(1)}_k \in \seq_i(h^{(1)})$ and $h^{(2)}_k \in \seq_i(h^{(2)})$. Since $\obs_i(h^{(1)})=\obs_i(h^{(2)})$, we must have $h^{(1)}_k,h^{(2)}_k \in I_k$ in $\Gamma$. Moreover,  $\obs_i(h^{(1)}_k)=\obs_i(h^{(2)}_k)$ since these are subsequences of $\obs_i(h^{(1)})$ and $\obs_i(h^{(2)})$, respectively, which are equal. Hence, by \Cref{def:prr}, $h^{(1)}_k\sim h^{(2)}_k$ and they must be in the same infoset in $\pr_i(\Gamma)$. Since this is true for all $k \in [\ell]$, this implies $\obs_i^*(h^{(1)})=\obs_i^*(h^{(2)})$, as desired. This proves $i$ has perfect recall in $\pr_i(\Gamma)$

    Next, take any game $\Gamma'$ such that $\Gamma' \succeq_i \Gamma$ and $i$ has perfect recall in $\Gamma'$. We would like to show that $\Gamma' \succeq_i \pr_i(\Gamma)$. For $\Gamma$, $\Gamma'$, $\pr_i(\Gamma)$, say $\calI_i$,$\calI_i'$, and $\calI_i^*$ are the infosets of $i$, respectively, and $\obs_i(h),\obs'_i(h)$, and $\obs^*_i(h)$ are as defined in \Cref{def:imperf_recall}, respectively. Take any $I^* \in \calI^*_i$. By \Cref{def:prr}, there exists a $I \in \calI_i$ and a $h\in I$ such that $I^*= \{h^*\in I: \obs_i(h^*) =\obs_i(h)\}$. Since $\Gamma' \succeq_i \Gamma$, by \Cref{def:ref/coarse}, there exists a $\mathcal{J}'_i \subseteq \calI'_i$ such that $I= \bigsqcup_{I'\in \mathcal{J}'_i} I'$. In particular, given any $h^* \in I^*$, there is a $I'\in \mathcal{J}'_i$ such that $h^* \in I'$. Since $i$ has perfect recall in $\Gamma'$, for any $h' \in I'$, we must have $\obs_i'(h')= \obs_i'(h^*)$, moreover, since infosets in $\calI'_i$ parition those in $\calI_i$, this implies  $\obs_i(h')= \obs_i(h^*)$, and hence $h' \in I^*$. This implies $I' \subseteq I^*$. Since $I' \in \mathcal{J}'_i$ was chosen as the infoset containing an arbitrary element $h^* \in I^*$, this implies there exists a $\mathcal{K}'_i \subseteq \mathcal{J}'_i$ such that $I^* = \bigsqcup_{I' \in \mathcal{K}_i}I'$, hence proving that $\Gamma' \succeq_i  \pr_i(\Gamma)$.
\end{proof}
\subsection{On the relationship between imperfect recall and imperfect information}
In the main body of the paper, we stated that any imperfect information game can be turned into a strategically-equivalent one with only imperfect recall by adding dummy nodes. Here, we formalize this:
\begin{prop}
    Given any game $\Gamma$ and  a player $i \in N$, there exists an imperfect-recall game $\Gamma'$ such that $u_1(\SC(\Gamma))=u_1(\SC(\Gamma'))$ for  $\SC \in \{\wcdt, 
\bcdt,\wcdtnash, 
\bcdtnash,$ $\wedt, \bedt,\wedtnash,$ $\bedtnash, \wnash, \bnash\}$, and $i$ has perfect information in $\pr_i(\Gamma')$ (\emph{i.e.}, each node is in a infoset of size 1).
\end{prop}
\begin{proof}
Given $\Gamma$ with nodes $\calH$, we define a new game $\Gamma'$ with nodes $\calH'=\calH \cup \mathcal{D}_i$ where $\mathcal{D}_i= \{d_h\}_{h \in \calH_i}$ is a set of new internal nodes, each belonging to player $i$ (\emph{i.e.}, $\calH_i' = \calH_i \cup \mathcal{D}_i$ and $\calH_j' = \calH_j$ for all $j \neq i$), with $|A_{d_h}|=1$ for each $d_h \in \mathcal{D}_i$. The utilities and game tree for $\Gamma'$ is identical to that of $\Gamma$, except each $h\in \calH_i$ is preceded by $d_h \in \mathcal{D}$, which has a single action leading to $h$. The infosets of all (original) nodes in $\Gamma'$ are the same as their infosets in $\Gamma$, and each node $d_h \in \mathcal{D}_i$ is in a infoset of size 1, say $I'_{d_h}$. Since $i$ has no choice in the nodes in $\mathcal{D}_i$, and since each player (including $i$) has the same information in both games in each of its decision nodes (with $|A_h|>1$), it is clear that $u_1(\SC(\Gamma))=u_1(\SC(\Gamma'))$ for any $\SC$ listed in the proposition statement. Moreover, for any distinct $h,h' \in \calH$, $\seq'_i(h) \neq \seq'(h')$ in $\Gamma'$, since one contains $I_{d_h}$ and the other contains $I_{d_{h'}}$. Therefore, each $h \in \calH'_i$ will be placed in a different infoset in $\pr_i(\Gamma')$, giving $i$ perfect information. 
\end{proof}

\subsection{Proof of \Cref{prop:recall_bad}}
We now prove that one player getting perfect recall can arbitrarily hurt every player.

\recallbad*
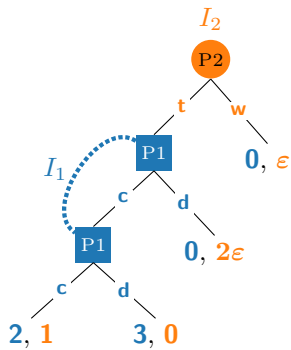
\begin{figure}[t]
    \tikzset{
        every path/.style={-},
        every node/.style={draw},
    }
    \forestset{
    subgame/.style={regular polygon,
    regular polygon sides=3,anchor=north, inner sep=1pt},
    }
    \centering
        \begin{forest}
            [\scriptsize{P2},p2gs,name=p20,s sep=25 pt,l sep=21pt 
                [\scriptsize{P1},p1gs,name=p0,el={2}{t}{},s sep=25pt,l sep=21pt
                    [\scriptsize{P1},p1gs,name=p1b,el={1}{c}{},s sep=25pt,l sep=21pt
                        [\util1{2}\text{, }\util2{1},terminal,el={1}{c}{},yshift=-3.3pt]
                        [\util1{3}\text{, }\util2{0},terminal, el={1}{d}{},yshift=-3.3pt]
                    ]
                    [\util1{0}\text{, }\util2{$\boldsymbol{2\varepsilon}$},terminal,el={1}{d}{},yshift=-3.3pt]
                ]
                [\util1{0}\text{, }\util2{$\boldsymbol{\varepsilon} $},terminal,el={2}{w}{},yshift=-3.3pt]
            ]\node[above=0pt of p20,draw=none,p2color]{$I_2$};
            \draw[infoset1] (p0) to [bend right=90,]node[left,draw=none,p1color,]{$I_{1}$} (p1b);
        \end{forest}
    \caption{A game with imperfect recall. Giving Player 1 (\ponegs) perfect recall hurts both players. Terminals show utilities for Player 1 and Player 2 (\ptwogs). Infosets are joined by dotted lines (repeated from \Cref{fig:recall_bad}).
    \label{fig:recall_only_bad}}
\end{figure}

\begin{proof} Say $\Gamma$ is the two-player game from \Cref{fig:recall_bad} (repeated as \Cref{fig:recall_only_bad}) with $\varepsilon  \in (0,1)$. Since Player 2 has a single decision node (and thereby a single infoset), the CDT/EDT/Nash equilibria of the single-player game from its perspective (for any fixed strategy of Player 1) coincide. Say $\pi$ is a CDT equilibrium of $\Gamma$ with $\pi_1(\bluec\mid I_1)=p$. Then playing $\oranget$ would bring P2 a utility of $p^2+(1-p)2\varepsilon$, whereas playing $\orangew$ bring a utility of $\varepsilon$. Since $p^2+(1-p)2\varepsilon > \varepsilon$ for all $p \in [0,1]$ and $\varepsilon \in (0,1)$, in order for $\pi$ to be a KKT point, we must have $\pi_2(\oranget \mid  I_2)=1$. Then, $u_1(\pi) = 2p^2 +3p(1-p)$. The only KKT point of this is $p=1$. Hence, the only CDT equilibrium of $\Gamma$ is for P1 and P2 to always play $\bluec$ and $\oranget$, bringing them utility $2$ and $1$, respectively.

Now consider $\pr_1(\Gamma)$, and say $\pi^*$ is a CDT equilibrium. Assume for the sake of contradiction that  $\pi^*_2(\oranget \mid  I_2)=q>0$. In that case, $u_1(\pi^*)= q(2p_1p_2+3p_1(1-p_2))$ where $p_1,p_2$ are the probabilities that $\pi^*_1$ places to $\bluec$ in the first and second decision node of P1, respectively. For any $q>0$, the only KKT point of this is $p_1=1$ and $p_2=0$. However, this implies $u_2(\pi^*)= (1-q)\varepsilon$, the only KKT point of which is $q=0$, which is a contradiction. Hence, we must have $\pi^*_2(\oranget \mid  I_2)=0$. This implies that the only CDT equilibrium of $\pr_1(\Gamma)$ (up to realization equivalence) is P2 always playing $\orangew$, bringing P1 and P2 utilities 0 and $\varepsilon$, respectively.

By \Cref{lem:EQ hierarchy}, this shows that  $\vor^{\SC}(\Gamma)=\frac{u_1(\SC(\pr_1(\Gamma)))}{u_1(\SC(\Gamma))}=\frac{0}{2}=0$ and $\frac{u_2(\SC(\pr_1(\Gamma)))}{u_2(\SC(\Gamma))}=\frac{\varepsilon}{1}=\varepsilon$ for any $\SC \in \{\wcdt, 
\bcdt, \wedt, \bedt, \wnash, \bnash\}$. In particular, since $\varepsilon$ can be arbitrarily close to 0, this proves the proposition.
\end{proof}

\subsection{Proof of \Cref{thm:hardness}}
We now prove that VoR is hard to compute. As stated in the main body of the paper, for this section alone, we assume that $u_i(z)\geq \eta$ for all $i \in \calN$ and $z \in \calZ$ for some $\eta >0$, to ensure VoR is bounded.

\vorhardness*

The proof of the theorem mostly relies on existing hardness results for computing some of these solution concepts in single-player imperfect-recall games. We fill first prove a novel hardness result for computing $\wcdt$ and $\wedt$. For $\wcdt$, even \NP-hardness is novel; for $\wedt$, while \NP-hardness was known~\citep{Tewolde23:Computational}, we prove a stronger inapproximability result.

\begin{lemma}\label{lemma:wedt-hard}
        For a single-player game $\Gamma$, both $u_1(\wcdt(\Gamma))$ and $u_1(\wedt(\Gamma))$ are \NP-hard to compute and \NP-hard to approximate to any multiplicative factor.
\end{lemma}
\begin{proof}
    Fix any $M\geq 1$. We will prove that the worst CDT and EDT equilibria utilities of a single-player game is \NP-hard to approximate to a multiplicative factor of $M$. We will be reducing from 3SAT. Let $x_1, \ldots , x_\ell$ be the variables of a 3CNF formula $\phi$ with $n$ clauses. We construct a game instance $\Gamma$ as follows. Each variable $x_i$ has a corresponding info set $I_i$ with $A_{I_i}=\{T,F\}$ (two actions for each info set). The root of the tree is a player node $h_0 \in \calH_1$, which is in an infoset of size 1 (say $I_0 \coloneqq\{h_0\}$) and $A_{h_0}=\{Y,N\}$. Playing $N$ at $h_0$ leads to a leaf node that brings utility $\eta$. Playing $Y$ at $h_0$, on the other hand, leads to chance node that uniformly at random selects one of $n$ subtrees, each corresponding a clause $C$ in $\phi$. The subtree associated with clause $C$ starts at a node $h_C$ with $|A_{h_C}|=1$ and $\{h_C\} \in \calI_1$ (\emph{i.e.}, $h_C$ is in its own infoset). The single action at $h_C$ leads to a binary tree of depth 3. Say the clause $C$ contains variables $x_i , x_j , x_k$; then, the nodes in the first, second, and third layer of the binary tree belong to infosets $I_i,I_j$, and $I_k$, respectively. Accordingly, each leaf follows a sequence of 3 actions that can be interpreted as a truth assignment to $x_i , x_j , x_k$. If this assignment satisfies the clause $C$, then the leaf node brings utility $\eta$. Otherwise, the leaf node brings utility $M' + \eta$, where $M'=8 \cdot M \cdot \eta \cdot n$. Since $\Gamma$ has no absentmindedness (no infoset is entered multiple times on the path to a leaf), its CDT and EDT equilibria coincide (\Cref{rem:edt = cdt w/o abs}).

    First, we argue that if $\phi$ is satisfiable, then the worst EDT equilibrium of $\Gamma$ yields utility $\eta$. Assume $\phi$ is satisfiable for truth assignments $\{x^{*}_i\}_{i \in [\ell]}$. Consider the strategy $\pi$ that at infoset $I_i$ plays the action $x^*_i$ with probability 1 for each $i \in [\ell]$, and plays $N$ with probability 1 at $h_0$ (all other nodes have a single action and hence a single strategy), bringing utility $\eta$. To show that this is an EDT equilibrium, we must show that the player cannot increase its utility by deviating from $\pi$ at a single infoset. Notice that since each $I \in \calI_1 \setminus \{I_0\}$ is not on the path of play (\emph{i.e.}, is visited with probability 0), any devation at $I$ cannot possibly increase the expected utility. A deviation at $I_0$ cannot increase the utility either: since all clauses are satisfied, $\pi$ achieves $\eta$ on each subtree under the chance node, so playing $Y$ at $h_0$ (or any mix between $Y$ and $N$) would also bring utility $\eta$. Since $\eta$ is the smallest utility in the game, $\pi$ is also a worst EDT equilibrium.

We next claim that if $\phi$ is not satisfiable, then for any strategy $\pi \in S_1$, we have $u_1(\pi\mid h_c)\geq (1+M)\eta$, where $h_c$ is the chance node in $\Gamma$. Fix any strategy $\pi$. If $\pi$ is pure, then it corresponds to a truth assignment for the variables of $\phi$. Since $\phi$ is not satisfiable, at least one clause needs to be not satisfied, and hence the corresponding subtree has expected utility $M'+\eta$. Starting from $h_c$, this tree is entered $\frac{1}{n}$ of the time, so $u_1(\pi\mid h_c)\geq \eta+ \frac{M'}{n}>(1+M)\eta$. Otherwise, say $\pi$ is mixed. Consider the pure strategy profile $\pi'$ constructed by rounding the probabilities in $\pi$, \emph{i.e.}, at an info set $I_i$, $\pi'$ always plays $T$ if $\pi$ plays $T$ with probability $\geq 0.5$, otherwise $\pi'$ always plays $F$. Since $\pi'$ is pure, it corresponds to a truth assignment, which needs to leave at least one clause $C$ unsatisfied. Say $x_i,x_j,x_k$ are the variables in $C$ and say $\pi$ plays the action $\pi'$ plays at $I_\alpha$ with $1/2 + \varepsilon_\alpha$ probability ($\varepsilon_\alpha\geq 0$) for $\alpha \in \{i,j,k\}$. If $h$ is the root of the subtree corresponding to $C$,
    \begin{align*}
        u_1(\pi\mid h) =& \eta+ M' \cdot \Prob[\pi\text{ plays an assignment not satisfying }C]\\
        \geq & \eta+  M'\cdot\Prob[\pi\text{ plays the outcome of } \pi']\\
        =& \eta+M' \cdot \left(\frac{1}{2}+ \varepsilon_i\right)\cdot  \left(\frac{1}{2}+ \varepsilon_j\right)\cdot  \left(\frac{1}{2}+ \varepsilon_k\right)\geq \eta + \frac{M'}{8}.
    \end{align*}

Starting from $h_c$, since the subtree corresponding to $C$ is entered with probability $\frac{1}{n}$ and since all leaves bring you at least $\eta$ utility,  $u_1(\pi\mid h_c) \geq \eta+ \frac{M'}{8n}= (1+M)\eta>\eta$. This implies that for $\pi$ to be an EDT equilibrium, it must play $Y$ with probability 1 at $I_0$, otherwise the player can increase its utility by deviating to always playing $Y$. Since $\pi$ was arbitrarily chosen, this implies that the worst EDT equilibrium utility is at least $(1+M)\eta$. 

Since the ratio between $u_1(\wedt(\Gamma))$ when $\phi$ is satisfiable and when it is not satisfiable is at least $\frac{(1+M)\eta}{\eta}=1+M$, approximating $u_1(\wedt(\Gamma))$ to a multiplicative factor of $M$ would allow distinguishing between these two cases, and hence determining whether $\phi$ is satisfiable, solving the 3SAT instance. Since 3SAT is \NP-hard and $M$ was arbitrarily chosen, this proves that it is  \NP-hard to approximate $u_1(\wedt(\Gamma))$ to any multiplicative factor. Since $\wedt$ and $\wcdt$ of $\Gamma$ coincide, the same is true for $u_1(\wcdt(\Gamma))$. Naturally, this also proves $u_1(\wedt(\Gamma))$ and $u_1(\wcdt(\Gamma))$ are \NP-hard to compute exactly.
\end{proof}

We now present the proof of the theorem.

\begin{proof}[Proof of \Cref{thm:hardness}]
    
    Consider a single-player game $\Gamma$.  Recall that  $\vor^{\SC}(\Gamma)=\frac{u_1(\SC(\pr_1(\Gamma)))}{u_1(\SC(\Gamma))}$. We will argue that the nominator is easy to compute for each $\SC$ in the theorem statement. For $\SC \in \{\bcdt,\bedt,\wnash,\bnash\}$, this is true since each of these solution concepts corresponds to the optimal play in a single-player perfect-recall game, which can be computed in polynomial time~\citep{Stengel96:Efficient}. For $\SC\in \{\wcdt, \wedt\}$, we have to be more careful, since these do not necessarily coincide with the optimal strategy, even with perfect recall, as seen in \Cref{ex:refinement_needed}. To sidestep this issue, we observe that in the class of games described in \Cref{lemma:wedt-hard}, it is straightforward to compute $u_1(\pr_1(\SC(\Gamma)))$ for $\SC \in \{ \wcdt, \wedt \}$. Indeed, say $\Gamma$ is the game from the construction of the proof of~\Cref{lemma:wedt-hard}. Because of the dummy nodes $h_C$ that observe the outcome of the chance node in $\Gamma$, in $\pr_1(\Gamma)$ the player will have separate information sets for each subtree under the chance node $h_c$. Consider then the strategy $\pi$ that at each subtree acts specifically to satisfy the clause associated with that subtree, obtaining $u_1(\pi\mid h_c)=\eta$. Moreover, say $\pi(N\mid I_{0})=1$, so the player never reaches the chance node. Since no infoset $I \in \calI_1 \setminus \{I_0\}$ is on the path of play, the player cannot increase its expected utility by deviating at $I$. Similarly, deviating at $I_0$ cannot bring positive utility, as both actions bring $\eta$ utility in expectation. Also, there is no absentmindedness, so CDT and EDT equilibria coincide. Hence, $u_1(\wedt(\Gamma))=u_1(\wcdt(\Gamma))=\eta$ for this class of games.
    
    Moreover, \citet{Tewolde23:Computational} show that $u_1(\SC(\Gamma))$ is \NP-hard to compute for single-player imperfect-recall games for $\SC \in \{\bcdt,\bedt,\wnash,\bnash\}$, and that the problem admits no $\FPTAS$ unless \NP=\ZPP. Since $u_1(\SC(\Gamma))$ can be computed by first computing $u_1(\SC(\pr_1(\Gamma)))$ and $\vor^{\SC}(\Gamma)$ (both of which are bounded and nonzero by the assumption on utilities) and getting their ratio, these same hardness results translate to $\vor^{\SC}(\Gamma)$. Similar reasoning yields the claimed hardness results with respect to $\SC \in \{\wcdt, \wedt \}$ by~\Cref{lemma:wedt-hard}. In these cases, the inapproximability results are stronger: it is \NP-hard to approximate $\vor^{\SC}(\Gamma)$ to any multiplicative factor.
\end{proof}

\subsection{Proof of \Cref{prop:1p_optimal_lowerbound}}
We show that recall can only help in single-player games in terms of optimal strategies and equivalent solution concepts.
\optvorgood*
\begin{proof}
Take any single-player game $\Gamma$ and say $\pi$ is an optimal strategy, \emph{i.e.} $u_1(\pi)=u_1(\opt(\Gamma))$. By \Cref{def:prr}, any infoset $I \in \calI_1$ of $\Gamma$ is partitioned into (possibly multiple) infosets $\mathcal{J}_I$ in $\pr_1(\Gamma)$. Consider a strategy $\pi'$ in $\pr_1(\Gamma)$ such for each $I \in \calI_1$ and each $I' \in \mathcal{J}_I$, we have $\pi'(\cdot \mid I')=\pi(\cdot \mid I)$, \emph{i.e.}, $\pi'$ acts at each $I' \in \mathcal{J}_I$ as $\pi$ acts in $I$. Clearly, $\pi'$ achieves the same utility as $\pi$, implying $u_1(\opt(\pr_1(\Gamma)))\geq u_1(\opt(\Gamma))$. Since optimal strategies coincide with the best CDT and best EDT equilibria of single-player games (see \Cref{lem:EQ hierarchy}), this proves the proposition statement.
\end{proof}
\subsection{Proof of \Cref{prop:1p_vor1}}

We next proceed with the proof of~\Cref{prop:1p_vor1}. (The proof of~\Cref{prop:edt-nash} is deferred to \Cref{sec:cdt-nash}, where we introduce some further background on CDT-Nash equilibria.)

\optone*

\begin{proof}
    For any $z\in \calZ$, $\seq(z)$ only contains nodes where the player acts, and each $I \in \calI_1$ appears at most once in $\obs(z)$. Playing $a_k$ with probability 1 in each $I_k$ guarantees reaching $z$ with probability 1. Therefore, the player can achieve utility $\max_{z \in \calZ} u_1(z)$ with a pure strategy, which is the max possible utility in both $\Gamma$ and $\pr_1(\Gamma)$, hence an optimal strategy. 
\end{proof}

\subsection{Proof of \Cref{lemma:am_strategy}}
We now prove a connection between the absentmindedness coefficient of a leaf and the probability of reaching that leaf in $\Gamma$.
\eminne*

\begin{proof}
    Given $\obs(z)=(i_k, I_k,a_k)_{k=0}^{\depth(z)-1}$, for each $I \in \calI_1$ and $a \in I$, define $\pi _z(a\mid I) = p_z(a)$ if $\exists k$ such that $I_k = I$, and pick an arbitrary mixed action for all other $I$. Since $\sum_{a \in A_I}n_z(a)=n_z(I)$ and hence $\sum_{a \in A_I}p_z(a)=1$, this is a valid strategy. Since there are no chance nodes, we have
    \begin{align*}
        \Prob(z \mid \pi _z) &= \prod_{k=0}^{\depth(z)-1} \pi _z(a_k\mid I_k) 
        = \prod_{k=0}^{\depth(z)-1} p_z(a_k)\\&= \prod_{\substack{{I \in \calI_1: n_z(I)>1}  \\{a \in A_I: n_z(a)>0}}} p_z(a)^{n_z(a)}= \am(z).
    \end{align*}
    $u_1(\pi _z) \geq \am(z)u(z)$ follows since the utilities are nonnegative. 
\end{proof}

\subsection{Proof of \Cref{prop:1p_am} and bounding $\am(z)$}

We prove the bound on single-player games with absentmindedness but no chance nodes.
\ambound*
\begin{proof}
    By \Cref{prop:1p_vor1}, there exists a pure optimal strategy in $\pr_1(\Gamma)$. Since there are no chance nodes in $\Gamma$, P1 is able to reach any $z \in \calZ$ in $\pr(\Gamma)$, so $u_1(\opt(\pr(\Gamma)))=\max_{z \in \calZ} u_1(z)=u_1(z^*)$. For $\Gamma$, by Lemma \ref{lemma:am_strategy}, we have $u_1(\opt(\Gamma)) \geq \max_{z \in \calZ} u_1(\pi _z) \geq \max_{z \in \calZ} \am(z)u_1(z) \geq \am(z^*)u(z^*)$. Plugging these into the definition for $\vor^{\opt}(\Gamma)$ gives us the first and second inequalities from the proposition statement. 
\end{proof}

To give intuition on how small $\am(z)$ can get, we prove a separate proposition that formalizes the worst-case scenario in terms of absentmindedness: only one leaf node brings positive utility, and reaching it requires playing each action equally often.
\begin{prop}\label{prop:am_degree}
    For each $z \in \calZ$, we have
    \begin{align*}
        \frac{1}{\am(z)} \leq  \prod_{I \in \calI_1: n_{z}(I) >1 } \min(n_{z}(I), |A_I|)^{n_{z}(I)}
    \end{align*}
\end{prop}
\begin{proof}
    For any $z\in \calZ$, we have
    \begin{align*}
        \log \am(z) &= \sum_{\substack{{I \in \calI_1: n_z(I)>1}  \\{a \in A_I: n_z(a)>0}}} n_z(I) p_z(a)\log(p_z(a))= \sum_{I \in \calI_1: n_z(I)>1}  -n_z(I)H(\pi _z(\cdot \mid I))
    \end{align*}
    where $H(\pi _z(\cdot \mid I))$ is the Shannon entropy of $\pi _z(\cdot \mid I)$. By construction, the size of the support of $\pi_z$ is bound by $\min(n_{z}(I), |A_I|)$. Since the entropy-maximizing (discrete) distribution is the uniform distribution, we have $H(\pi _z(\cdot \mid I)) \leq \log(\min(n_{z}(I), |A_I|))$ and hence
    \begin{align*}
         \log \am(z) &\geq \sum_{I \in \calI_1: n_z(I)>1}  -n_z(I)\log(\min(n_{z}(I), |A_I|))\\
         &= \log \left(\prod_{I \in \calI_1: n_{z^*}(I) >1 } \min(n_{z^*}(I), |A_I|)^{-n_{z^*}(I)}\right)\\
         \Rightarrow  \am(z) &\geq  \left(\prod_{I \in \calI_1: n_{z^*}(I) >1 } \min(n_{z^*}(I), |A_I|)^{n_{z^*}(I)}\right)^{-1},
    \end{align*}
    giving us the inequality of the lemma.
\end{proof}
The game from \Cref{ex:1am_tight} meets this upper bound, since $\frac{1}{\am(z^*)}= 2^n = |A|^{n_{z^*}(I)}$, showing that it is indeed tight.
\subsection{Proof of \Cref{cor:1pm_am}}
In the main paper of the body, we gave a corollary of \Cref{prop:1p_am} that relates to game classes.
\amclass*
\begin{proof}
        $\vor^{\opt}(\class) \leq \max_{z \in \calZ} \frac{1}{\am(z)}$ follows from \Cref{prop:1p_am}. To prove the lower bound, consider a game $\Gamma'$ that has the same game tree and infoset partition as $\Gamma$, but with different utilities: $u_1(z^*)=1$ for some $z^* \in \argmin_{z \in \calZ} \am(z)$, and $u_1(z)=0$ for all $z \in \calZ \setminus \{z^*\}$. We indeed have $\vor^\opt(\Gamma')=\frac{1}{\alpha(z^*)}$. Since $\Gamma' \in \class$, this proves that $\vor^{\opt}(\class) \geq \max_{z \in \calZ} \frac{1}{\am(z)}$.
\end{proof}

\subsection{Proof of \Cref{prop:1p_chance}}
Before proving \Cref{prop:1p_chance}, we first present an existing result about games without absentmindedness.
\begin{lemma}[\citealp{Koller92:Complexity}]\label{lemma:pure_optimal}
    For a single-player game without absentmindedness, there exists a pure optimal strategy.
\end{lemma}

We also prove two additional preliminary lemmas:
\begin{lemma}\label{lemma:pure_chance}
Given a single-player game $\Gamma$ and a pure strategy $\pi$, say $\{z_1,z_2,\ldots,z_\ell\} \subseteq \calZ$ are the leaves that the player reaches with nonzero probability under $\pi$. Then $u_1(\pi)= \sum_{i=1}^\ell \chance(z_i)u_1(z_i)$ and $ \sum_{i=1}^\ell \chance(z_i)=1$.
\end{lemma}
\begin{proof}
    Fix any $i \in [\ell]$ and consider $\obs(z_i)=(i_k,I_k,a_k)_{k=0}^{\text{depth}(z_i)-1}$. Since $\pi$ is pure, for all $k$ such that $i_k=1$, we must have $\pi(a_k\mid I_k)=1$, otherwise $z_i$ would have 0 reach probability. Hence, if $k_1,k_2, \ldots k_m$ are the steps that correspond to chance nodes, \emph{i.e.}, $i_{k_j} =c$ for all $j \in [m]$, we have
    \begin{align*}
        \Prob(z_i | \pi)  = \prod_{j=1}^m \Prob_c(a_{k_j} | h_{k_j}) = \chance(z_i).
    \end{align*}
    Hence, $u_1(\pi)= \sum_{i=1}^\ell  \Prob(z_i | \pi) u_1(z_\ell)= \sum_{i=1}^\ell   \chance(z_i) u_1(z_\ell)$ and
    $\sum_{i=1}^\ell \chance(z_i) = \sum_{i=1}^\ell  \Prob(z_i | \pi)=1$.
\end{proof}
\begin{lemma}\label{lemma:betabound}
    Given a single-player game $\Gamma$ and a pure strategy $\pi$, pick any chance node $h \in \calH_c$. If $\Prob(h|\pi)>0$, then $|\{z \in \calZ: h \in \seq(z), \Prob(z|\pi)>0\}| \leq \branch(h)$. In words, the number leaf nodes in the subtree rooted at $h$ with nonzero reach probability is at most $\branch(h)$.
\end{lemma}
\begin{proof}
Recall that $\branch(h)=\sum_{a \in A_h}b_h(a)$, where
\begin{align*}
b_h(a) = \begin{cases} 1 & \text{if }|H_{ha}|=0\\ \max_{h \in H_{ha}}\branch(h)& \text{otherwise}
\end{cases},
\end{align*}
where $H_{ha} \subset \calH_c$ are the chance nodes in the subtree rooted at $ha$. We prove the lemma by induction on the number of chance nodes under $h$. For the base case, assume that the subtree rooted at $h$ contains no chance nodes. Then for each $a \in A_h$ such that $\Prob_c(a \mid h)>0$, the subtree rooted at $ha$ will have exactly one leaf node with nonzero reach probability, since $\pi$ is pure. As such, the number of such leaf  nodes in the subtree rooted at $h$ is at most $|A_h|$ (which is acheived if it puts nonzero probability in all of its actions), which is exactly $\branch(h)$ since $b_h(a)=1$ for all $a \in A_h$. For the inductive step, assume that the lemma statement is true for all chance nodes that contain at most $k-1$ chance nodes in their subtree. Say $h \in \calH_c$ contains $k$ chance nodes in its subtree. Fix any $a \in A_h$. Notice that since $\pi$ is pure, there can be at most one chance node in $h' \in H_{ha}$ such that $\Prob(h'\mid \pi)>0$ and $h$ is the latest chance node in $\seq(h')$. If there is no such node, then there can be at most one leaf in the subtree of $ha$ with nonzero reach probability. If there exists such a $h'$, all of the leaves under $ha$ with nonzero reach probability needs to be in the subtree rooted at $h'$ (if there exists a $z$ with nonzero reach probability under $ha$ but not $h'$, this would imply $\pi$ played a mixed strategy at the earliest split between $\seq(z)$ and $\seq(h')$, which is a contradiction). Since $h'$ has at most $k-1$ chance nodes under it, by the inductive hypothesis this implies that there are at most $\branch(h') \leq \max_{h'' \in H_{ha}} \branch(h'')$ leaves under $ha$. Summing over all $a\in A_h$, this implies that the total number of leaves under $h$ with nonzero reach probability is upper bounded by $\sum_{a \in A_h} \max(1, \max_{h' \in H_{ha}} \branch(h')) =\branch(h)$, as desired.
\end{proof}
We now turn to proving the proposition.

\chancebound*

\begin{proof}
For any $z \in \calZ$ with $\obs(z)=(i_k,I_k,a_k)_{k=0}^{\depth(z)-1}$, the strategy $\pi _z$, which plays $a_k$ with probability 1 for each non-chance node, achieves an expected utility of at least $\chance(z)u_1(z)$ since the utilities are nonnegative; hence $u_1(\opt(\Gamma))\geq \max_{z \in \calZ} \chance(z)u_1(z)$, giving the first inequality. For the second inequality, say $\pi ^*$ is a pure optimal strategy in $\pr_1(\Gamma)$, which exists by \Cref{lemma:pure_optimal}. Say $\{z_1,z_2,\ldots,z_\ell\} \subseteq \calZ$ are the leaves that the player reaches with nonzero probability under $\pi^*$. By \Cref{lemma:pure_chance}, we have $u_1(\pi^*)= \sum_{i=1}^\ell \chance(z_i)u_1(z_i)$ and $ \sum_{i=1}^\ell \chance(z_i)=1$. With imperfect recall, the player can follow $\pi_{z^\dagger}$ for $z^\dagger \in \argmax_{j \in [\ell]} \chance(z_j)u_1(z_j)$. By an averaging argument, this quantity is at least $\frac{u_1(\pi^*)}{\ell}$, ensuring $\vor(\Gamma) \leq \ell$. Say $h \in \calH^c$ is the first chance node that $\pi^*$ enters (otherwise, $\ell=1$, so $\vor^{\opt}(\Gamma)=1$ since $z_1$ can be reached under imperfect recall too). By \Cref{lemma:betabound}, we have $\ell \leq \beta(h) \leq \max_{h' \in \calH^c} \beta(h)'$, completing the proof of the proposition.
\end{proof}

\subsection{Proof of \Cref{thm:1p}}
We now show that our two bounds from \Cref{prop:1p_am} and \Cref{prop:1p_chance} compose in games that have both absentmindedness and chance nodes.
\onethm*

\begin{proof}
Given $\Gamma$, say $\pi$ is a pure optimal strategy of $\pr_1(\Gamma)$, which exists by \Cref{lemma:pure_optimal}. Say $z_1, \ldots z_\ell$ are the leaves in $\Gamma$ with nonzero reach probability under $\pi$. By \Cref{lemma:pure_chance}, we have $u_1(\opt(\Gamma))=u_1(\pi)= \sum_{i=1}^\ell \chance(z_i)u_1(z_i)$ and $ \sum_{i=1}^\ell \chance(z_i)=1$. Moreover, for any $z \in \calZ$ with $\obs(z)=(i_k,I_k,a_k)_{k=0}^{\depth(z)-1}$, the strategy $\pi_z$ (as defined in the proof of \Cref{lemma:am_strategy}) has $\Prob(z\mid \pi_z)=\am(z)\chance(z)$ and achieves an expected utility of at least $\am(z)\chance(z)u_1(z)$ since the utilities are nonnegative; hence $u_1(\opt(\Gamma))\geq \max_{z \in \calZ} \am(z)\chance(z)u_1(z)$. This implies
\begin{align*}
    \vor^{\opt}(\Gamma)& \leq \frac{\sum_{i=1}^\ell \chance(z_i)u_1(z_i)}{ \max_{z \in \calZ} \am(z)\chance(z)u_1(z)} =  \sum_{i=1}^\ell \frac{ \chance(z_i)u_1(z_i)}{ \max_{z \in \calZ} \am(z)\chance(z)u_1(z)} \\&\leq \sum_{i=1}^\ell \frac{ 1}{\am(z_i)} \leq \max_{z \in \calZ} \frac{\ell}{\am(z)} \leq \max_{z \in \calZ, h \in \calH_c} \frac{\branch(h)}{\am(z)},
\end{align*}
where the last inequality follows from \Cref{lemma:betabound}. Since $\branch(h)$ and $\am(z)$ are both independent of utilities, this bound applies for all $\Gamma'$ that has the same game tree and infosets as $\Gamma$. Hence, if $\class$ is the class of games that share the same game tree and infosets as $\Gamma$, then
    \begin{align*}
       \vor^\opt(\class) \leq \max_{z \in \calZ, h \in \calH_c} \frac{\branch(h)}{\am(z)}.
    \end{align*}
\end{proof}

\subsection{An example of a smooth game}

Next, we provide an example of a \emph{smooth} game. In what follows, it will be convenient to work with the following extension of~\Cref{def:smooth}.

\begin{defn}[Extension of~\Cref{def:smooth}]
    \label{def:smooth-ext}
    Suppose that there exist functions $\Pi_I : S \to \mathbb{R}$ and $\Pi'_I : S_{-I} \to \mathbb{R}$ such that $\sum_{I \in \calI_1} \Pi_I(\pi) \leq u_1(\pi)$ and $\Pi_I(\pi) \propto u_1(\pi) + \Pi'_I(\pi_{-I})$ for every $I \in \calI_1$ and $\pi \in S$. A single-player game $\Gamma$ is \emph{$(\lambda, \mu)$-smooth} if there exists $\pistar \in S$ such that for any $\pi \in S$,
    \begin{equation*}
        \sum_{I \in \calI_1} \Pi_I(\pi^*_{I}, \pi_{-I}) \geq \lambda u_1(\opt(\Gamma)) - \mu u_1(\pi).
    \end{equation*}
\end{defn}

In particular, compared to~\Cref{def:smooth}, above we replaced the left-hand side of~\eqref{eq:smooth} with $\sum_{I \in \calI_1} \Pi_I(\pi^*_{I}, \pi_{-I})$; this is a relaxation as one can simply take $\Pi_I(\pi) \coloneqq u_1(\pi)/|\calI_1|$ and $\Pi'_I(\pi) \coloneqq 0$. It is easy to see that all implications of smoothness we saw earlier in the main body, and in particular~\Cref{prop:smooth}, extend under~\Cref{def:smooth-ext}. We are now ready to present an example of a smooth (imperfect-recall) game.

\begin{ex}\label{ex:IR_smooth}
    This example is based on \emph{valid utility games}~\citep{Vetta02:Nash}. Suppose that there is a ground set $E$ and a nonnegative, nondecreasing, submodular function\footnote{ We recall that a function $V: 2^E \to \mathbb{R}$ is submodular if $V( X \cap Y ) + V(X \cup Y) \leq V(X) + V(Y)$ for every $X, Y \subseteq E$.} $V$ defined on subets of $E$. We construct a single-player, imperfect-recall game $\Gamma$ as follows. At every infoset $I \in \calI_1$, the player selects an action $a_I \in A_I \subseteq 2^E$, whereupon the player forgets taking that action. The resulting utility is $V( U(a))$, where we use the notation $U(a) \coloneqq \bigcup_{I \in \calI_1} a_I$. Less abstractly, such games capture facility location problems orchestrated by profit-maximizing monopolies~\citep{Vetta02:Nash}. Further applications are discussed by~\citet{Goemans04:Market}.
    
    We will show that $\Gamma$ is $(1, 1)$-smooth per~\Cref{def:smooth-ext}. Let $a, a^* \in \bigtimes_{I \in \calI_1} A_I$ be two action profiles. By the properties of $V$,
    \begin{align}
        \label{align:ext-smooth}
        \sum_{ I \in \calI_1} [ V (U(a^*_I, a_{-I})) - V(U(\emptyset, a_{-I}))] \geq V(U( a^* )) - V(U(a));
    \end{align}
    this derivation is similar to~\citep[Example 2.6]{Roughgarden15:Intrinsic}. Now, we define $\Pi_I(a) \coloneqq V(U(a)) - V(U(\emptyset, a_{-I}))$ for each $I \in \calI_1$. For convenience, we let $\calI_1 \coloneqq \{I_1, \dots, I_n\}$ and $M_i(a) \coloneqq \bigcup_{j=i}^n a_{I_j}$. Then, we have
    \begin{align*}
        \sum_{ k = 1 }^n V(U( \emptyset, a_{-I_k} )) &= V(U(\emptyset, a_{-I_1})) + V(U(\emptyset, a_{-I_2})) + \sum_{k = 3 }^n V(U(a_{-I_k})) \\
        &\geq  V( U(a)) + V(M_3(a)) + \sum_{k = 3 }^n V(U(a_{-I_k})) \\
        &= V( U(a)) + V(M_3(a)) + V(U(a_{-I_3})) + \sum_{k = 4 }^n V(U(a_{-I_k})) \\
        &\geq 2 V(U(a)) + V(M_4(a)) + \sum_{k = 4}^n V(U(a_{-I_k})) \\
        &\geq \dots \\
        &\geq (n-2) V( U(a)) + V(M_{n}(a)) + V(U(a_{-I_n})) \\
        &\geq (n-1) V(U(a)),
    \end{align*}
    where we used the submodularity of $V$. In turn, this implies that 
    \[
        \sum_{I \in \calI_1} \Pi_I(a) = n V(U(a)) - \sum_{k = 1 }^n V(U(\emptyset, a_{-I_k})) \leq V(U(a)).
    \]
    As a result, the functions $\{ \Pi_I \}_{I \in \calI_1}$ satisfy the preconditions of~\Cref{def:smooth-ext}. Finally, considering $a^*$ to be a welfare-maximizing action profile and taking expectations in~\eqref{align:ext-smooth} yields the smoothness property with $\lambda = 1$ and $\mu = 1$, as claimed.
\end{ex}

\paragraph{Correlated solution concepts} It is worth noting that the primary motivation behind Roughgarden's smoothness was to provide PoA bounds not just for Nash equilibria, but also for more permissive equilibrium concepts that are computationally tractable; namely, \emph{(coarse) correlated equilibria}. This feature of smoothness is readily inherited by~\Cref{def:smooth}. More precisely, let us assume that the underlying game has no absendmindness. EDT equilibria can then be phrased as the Nash equilibria of a certain multi-player (normal-form) game (per~\Cref{def:EDT}), and it can be shown that they are hard to compute~\citep{Tewolde24:Imperfect}. On the other hand, one can define a relaxation thereof, say \emph{correlated EDT equilibrium}, as a correlated equilibrium of the corresponding game, which can be instead computed in polynomial time~\citep{Papadimitriou08:Computing,Jiang11:Polynomial}. A further compelling aspect of that equilibrium is that it arises when each infoset (separately) is consistent with the \emph{no-regret} property, which is satisfied by many natural learning algorithms~\citep{Cesa-Bianchi06:Predictiona}. \Cref{prop:smooth} applies even for that broader set of equilibria.

\subsection{Further connections}
Here, we discuss in more detail some previously introduced concepts that relate to the value of recall.

\paragraph{Price of uncorrelation in adversarial team games}

The \emph{price of uncorrelation} in adversarial team games was introduced by~\citet{Celli18:Computational} (see also~\citealp{Basilico17:Team}). In particular, in their setting, it is assumed that a team of players with identical interests competes against a single adversary---whose utility is opposite from the utility of the team. Among others, \citet{Celli18:Computational} compare the utility of the team in a \emph{team maxmin equilibrium (TME)}---the Nash equilibrium that maximizes the team's expected utility---to that in a \emph{team maxmin equilibrium with communication device (TMECom)}---in which players are able to interact and transmit information to a mediator. We observe that their result can be translated in terms of $\vor^{\bnash}(\class^{2p0s})$, where $\class^{2p0s}$ is the class of two-player zero-sum games (based on their construction, one of the players---corresponding to the adversary---has perfect recall). In particular, we have the following.

\begin{cor}[Consequence of \citealp{Celli18:Computational}]
    Let $\class^{2p0s}$ be the class of two-player zero-sum games with $|\calZ|$ terminal nodes. Then, $\vor^{\bnash}(\class^{2p0s}) \geq |\calZ|/2$.
\end{cor}

\paragraph{Price of miscoordination in security games}

The \emph{price of miscoordination} in multi-defender security games was introduced by \citet{Jiang13:Defender}. A (single-defender) security game involves two players (a defender and an attacker), a set of \emph{targets} $T$, and a set of \emph{resources} $R$, each of which has a set of feasible subsets of $T$ to which it can be assigned. The strategies of the defender is to pick these assignments. The strategies of the attacker is to pick a target to attack. The utilities to both players depend on the target that was attacked, and whether it was covered by one of the resources assigned by the defender. Importantly, the attacker picks a target after observing the (possibly mixed) strategy that the defender \emph{commits} to. As such, the attacker will always respond to a commitment with the strategy that maximizes its utility (\emph{i.e.}, \emph{best response} to the defender), and hence the defender will commit to the mixed strategy that maximizes its own gain given that the attacker will be best responding. The associated solution concept to this arrangement is a \emph{Stackelberg equilibrium}---for a formal definition, see \citet{Korzhyk11::Stackelberg}.

\citet{Jiang13:Defender} expand this setting to games with multiple defenders with identical interests, where different resources might be controlled by different defenders, who cannot correlate their strategies. Alternatively, their expansion can be seen as giving the defender imperfect recall (in the extensive form game where it assigns each resource one by one). Hence, their results on the price of miscoordination translate to value of recall bounds for (single-defender) security games, where the defender may have imperfect recall. In particular, while the value of recall can be unbounded in this class, we have the following bounds for the subclass of games with identical targets (in terms of the covered/uncovered utilities for the defenders and the attackers):
\begin{cor}[Consequence of \citealp{Jiang13:Defender}]
    Let $\class^{it}$ be the class of security games with identical targets, where the defender may have imperfect recall. Then, $4/3 \leq \vor^{\text{Stack}}(\class^{it}) \leq \frac{e}{e-1}$.
\end{cor}

\subsection{Proof of \Cref{thm:partial}}

We next proceed with the proof of~\Cref{thm:partial}. It is based on a reduction from \emph{exact cover by $3$-sets} (\XS). Here, we are given as input a universe $U \coloneqq \{1, 2, \dots, n\} $ and a collection of sets $ \mathcal{F} \coloneqq \{ F_1, \dots, F_m \}$ such that $F_i \subseteq U$ and $|F_i| = 3$ for all $i \in [m]$. The goal is to determine whether there is a subset $\mathcal{F}' \subseteq \mathcal{F}$ such that 
\begin{itemize}
    \item $F \cap F' = \emptyset $ for all $F, F' \in \mathcal{F}'$, and
    \item $\bigcup_{ F \in \mathcal{F}'} F = U$.
\end{itemize}

We assume that $n = 3m$. In the following proof, we will make use the fact that $\XS$ is \NP-hard~\citep{Garey79:Computers}.

\partialhard*

\begin{proof}
    Given as input an instance $\mathcal{P}$ to the $\XS$ problem, we construct a single-player (imperfect-recall) extensive-form game $\Gamma = \Gamma(\mathcal{P})$ as follows. Chance selects an element in the universe $U = \{1, \dots, n \}$ uniformly at random, which is subsequently observed by the player. Next, in each of the resulting infosets (each containing a single node), the player has a single (``dummy'') action, which now leads to the same infoset $I=\{h_u\}_{u \in U}$. That is, the player forgets the observation indicating which element of the universe was selected initially by the chance node. At each node $h_u \in I$, the player acts by selecting one of $|\mathcal{F}|$ actions, each corresponding to a set in $\mathcal{F}$, whereupon the game terminates. The utilities are then defined as follows. If the element $u \in U$ was originally drawn and the player selected the action corresponding to $F_j \in \mathcal{F}$, the utility is $1$ if $u \in F_j$ and $0$ otherwise. To complete this polynomial-time reduction, we let $k \coloneqq m-1$.

    We first show that if $\XS(\mathcal{P}) = 1$, then there is a $k$-partial recall refinement of $\Gamma$, with $k = m-1$, such that the utility of the player under an optimal strategy is $1$. Indeed, let $\mathcal{F}'$ correspond to a partition that solves $\XS(\mathcal{P})$. We then consider the partial recall refinement of $\Gamma$ in which for every $F \in \mathcal{F'}$, all nodes in $\{h_u\}_{u \in F}$ belong to a single infoset (say $I_F$) of their own; by the assumed property of $\mathcal{F}'$, it follows that there is such a $k$-partial recall refinement with $k = m-1$. We then consider the strategy in that $k$-partial recall refinement in which for every infoset $I_F$, corresponding to some set $F \subseteq U$, the player selects $F \in \mathcal{F}$. By construction, it follows that this strategy secures a utility of $1$, as claimed.

    Finally, we argue that if $\XS(\mathcal{P}) = 0$, then every strategy in any possible $k$-partial recall refinement of $\Gamma$ attains a utility (strictly) below $1$. Indeed, by construction of $\Gamma$, a player can secure a utility of $1$ iff only terminal nodes with utility $1$ have positive reach probability under its strategy. Consider a partition of $U$ corresponding to a $k$-partial recall refinement of $\Gamma$. Then, for any node $h_u$ for $u \in U$, it must be the case that the player assigns positive probability only to a subset $F \in \mathcal{F}$ such that $u \in F$. Since $\mathcal{F}$ only contains subsets with exactly $3$ elements, it follows that any $k$-partial recall refinement of $\Gamma$ with an optimal strategy securing utility $1$ must consist of infosets with exactly $3$ nodes. Further, it must also be the case that any such nodes form a set belonging to $\mathcal{F}$, and that the corresponding action is selected in an optimal strategy. This is only possible if $\XS(\mathcal{P}) = 1$.
\end{proof}
\section{Discussion on CDT and CDT-Nash}
\label{sec:cdt-nash}

Finally, this section provides further background on CDT and CDT-Nash equilibria, which was omitted from the main body due to space constraints. As we pointed out earlier (\Cref{rem:edt = cdt w/o abs}), CDT differs from EDT only in games with absentmindedness. To illustrate the difference, recall that player $i$ receives expected utility $u_i(\pi)$ from strategy profile $\pi$. EDT reasons that if player~$i$ deviates from $\pi$ at an infoset $I$ to a randomized action $\sigma \in \Delta(A_I)$, player $i$ can expect to receive the utility $u_i(\pi^{I \mapsto \sigma}_i, \pi_{-i})$ instead. CDT, on the other hand, argues that the player should instead expect the utility
\begin{align}
\label{appeq:cdt utils}
\begin{aligned}
    \U_i^{\cdt}&(\sigma \mid \pi, I) \coloneqq
    \U_i(\pi) + \sum_{a \in A_I} (\sigma(a) - \pi(a \mid I)) \cdot \nabla_{I,a} \, \U_{i}(\pi),
\end{aligned}
\end{align}
where $\nabla_{I,a} \, \U_i(\pi)$ denotes the partial derivative of $\U_i$ with respect to action $a$ and infoset $I$ at point~$\pi$. In particular, CDT utilities are first-order approximations of the nonlinear utility function $\U_i$. As such, and as we saw earlier in the main body, the CDT equilibrium conditions are characterized by admitting no first-order improvements compared to the utility currently achieved---formalized through the notion of a \emph{Karush-Kuhn-Tucker (KKT)} point~\citep{Boyd04:Convex}.

Let us next turn to equilibrium refinements. In the main body, we introduced the new notion of a EDT-Nash equilibrium; the definition of CDT-Nash equilibria, introduced by~\citet{Lambert19:Equilibria}, has been deferred until now; it is worth noting its resemblance with \emph{sequential equilibria}~\citep{Kreps82:Sequential}. We define $\Fr(I \mid \pi) := \sum_{h \in I} \Prob(h \mid \pi)$ to be the \emph{frequency} with which an infoset $I$ is \emph{visited} under strategy $\pi$ in game $\Gamma$.

\begin{defn}
\label{appdefn:cdt seq rat}
    A strategy $\pi$ in a single-player game $\Gamma$ is called \emph{CDT-\rat{}} if there is a sequence $(\pi^{(k)}, \eps^{(k)})_{k \in \N}$ such that
    \begin{enumerate}
        \item each $\pi^{(k)}$ is a strategy in $\Gamma$ such that $\pi^{(k)}(a \mid I) > 0$ for all $I$ and $a$, and $(\pi^{(k)})_{k \in \N}$ converges to $\pi$;
        \item each $\eps^{(k)} > 0$ and $(\epsilon^{(k)})_{k \in \N}$ converge to $0$; and
        \item for each $k$, for all $I$ with $\Fr(I \mid \pi^{(k)} ) >0$ and $\sigma \in \Delta(A_I)$,
        \begin{align*}
        \label{appeq:seq rat}
            \frac{1}{\Fr(I \mid \pi^{(k)} )} \cdot \Big( \U_1^{\cdt}&(\sigma \mid \pi^{(k)}, I) - \U_1(\pi^{(k)}) \Big) \leq \eps^{(k)}.
        \end{align*}
    \end{enumerate}
\end{defn}

We are now ready to introduce the definition of CDT-Nash, which mirrors the definition of EDT-Nash defined earlier in the main body.

\begin{defn}[\citealt{Lambert19:Equilibria}]
\label{appdefn:cdt nash}
    A profile $\pi$ is an \emph{CDT-\NE{}} of $\Gamma$ if it is a CDT equilibrium and if for all $i \in \calN$, and in the single-player perspective of $\Gamma$ (where every other player plays fixed $\pi_{-i}$), the strategy $\pi_i$ is realization-equivalent to a CDT-\rat{} strategy $\pi$. 
\end{defn}

In the context of our paper, a key property of CDT-Nash is that, under perfect recall, it agrees with Nash equilibrium~\citep{Lambert19:Equilibria}; this was stated in the main body as~\Cref{prop:cdt-nash}.

\subsection{Proof of \Cref{prop:edt-nash}}

We next proceed with the proof of~\Cref{prop:edt-nash}, the statement of which is recalled below.

\edtnash*

\begin{proof}
    The first part of the claim is clear by definition. For the remaining statements, we will use~\citep{Tewolde23:Computational}[Lemma 13; Appendix E]. To do so, we note that for a single-player game $\Gamma$, strategy $\pi$, infoset $I$, and deviation $\sigma \in \Delta(A_I)$,
\begin{enumerate}
    \item[(a)] $\frac{1}{\Prob(I \mid \pi)} \cdot \Big( \U_1(\pi^{I \mapsto \sigma}) - \U_1(\pi) \Big)$ in our notation translates to $$\textnormal{EU}_{\textnormal{EDT}, \textnormal{GDH}}(\sigma \mid \pi, I) - \textnormal{EU}_{\textnormal{EDT}, \textnormal{GDH}}\big( \pi(\cdot \mid I) \mid \pi, I \big)$$ in their notation; this holds because $\U_1(\pi^{I \mapsto \sigma}) - \U_1(\pi)$ simplifies to $$\displaystyle \sum_{z \in \calZ \text{ with } I \in \obs(z)} \Prob(z \mid \pi) \cdot u_1(z) - \sum_{z \in \calZ \text{ with } I \in \obs(z)} \Prob(z \mid \pi^{I \mapsto \sigma}) \cdot u_1(z).$$
    \item[(b)] Further, $\frac{1}{\Fr(I \mid \pi)} \cdot \Big( \U_1^{\cdt}(\sigma \mid \pi, I) - \U_1(\pi) \Big)$ in the notation of \Cref{appdefn:cdt seq rat} translates to $\textnormal{EU}_{\textnormal{CDT}, \textnormal{GT}}(\sigma \mid \pi, I) - \textnormal{EU}_{\textnormal{CDT}, \textnormal{GT}}\big( \pi(\cdot \mid I) \mid \pi, I \big)$ in their notation; this is a consequence of~\eqref{appeq:cdt utils} together with~\cite{Tewolde23:Computational}[Lemma~14 and (6)].
\end{enumerate}

As a result, it follows that for multi-player games without absentmindedness \Cref{defn:edt nash} and \Cref{appdefn:cdt seq rat} coincide (where we used that $\Prob(I \mid \pi) = \Fr(I \mid \pi)$ under that assumption); in turn, this implies that EDT-Nash and CDT-Nash coincide. Therefore, we can use \Cref{prop:cdt-nash} to conclude that in the special case of perfect recall, EDT-Nash coincides with Nash. This completes the proof.
\end{proof}

\end{document}